\documentclass[cleveref,thm-restate,numberwithinsect]{lipics-v2021}
\pdfoutput=1
\usepackage[T1]{fontenc}
\usepackage{amsmath,amsthm,amssymb}
\usepackage{thmtools}
\usepackage{graphicx}
\usepackage{enumerate}
\usepackage{algorithm}
\usepackage[noend]{algpseudocode}
\usepackage{todonotes}
\usepackage{dsfont}
\usepackage{xspace}
\usepackage{xcolor}
\usepackage{tikz}
\usepackage[mathlines]{lineno}
\linenumbers

\usetikzlibrary{patterns}
\usetikzlibrary{calc} 
\usetikzlibrary{decorations.pathreplacing}
\usetikzlibrary{math} 

\newcommand{\Alg}{\mathcal{A}}
\newcommand{\caA}{\mathcal{A}}
\newcommand{\NN}{\mathbb{N}}%
\newcommand{\Cc}{\mathcal{C}}%
\newcommand{\Ss}{\mathcal{S}}%
\newcommand{\Pp}{\mathcal{P}}%

\newcommand{\dd}{.\,.}
\newcommand{\Oh}{O}

\newcommand{\eps}{\varepsilon}%
\newcommand{\Ot}{\tilde{O}}
\DeclareMathOperator*{\poly}{poly}
\DeclareMathOperator*{\polylog}{polylog}
\newcommand{\ham}{\mathsf{hd}}
\newcommand{\hd}[2]{\ham( #1, #2)}
\newcommand{\hdk}[3]{\ham_{\le #3}( #1, #2)}
\newcommand{\enc}{\mathsf{enc}}

\newcommand{\ed}{\mathsf{ed}}
\newcommand{\edd}[2]{\ed( #1, #2)}
\newcommand{\edk}[3]{\ed_{\le #3}( #1, #2)}
\newcommand{\sk}{\ensuremath \mathsf{sk}^\ham}
\newcommand{\skk}[1]{\ensuremath \sk_k( #1)}
\newcommand{\ske}{\ensuremath \mathsf{sk}^{\ed}}
\newcommand{\skek}{\ske_k}
\newcommand{\LED}{\textsc{LED}\xspace}
\newcommand{\LHD}{\textsc{LHD}\xspace}
\newcommand{\PAL}{\textsc{PAL}\xspace}
\newcommand{\SQ}{\textsc{SQ}\xspace}
\newcommand{\kpalh}{$k$-\LHD-\PAL}
\newcommand{\kpale}{$k$-\LED-\PAL}
\newcommand{\ksqh}{$k$-\LHD-\SQ}
\newcommand{\ksqe}{$k$-\LED-\SQ}
\newcommand{\lcp}{\textsf{LCP}\xspace}
\newcommand{\occ}{\mathsf{Occ}}
\newcommand{\occj}{\occ_j}
\newcommand{\MI}{\mathsf{MI}}
\newcommand{\ES}{\mathsf{ES}}

\newcommand{\rot}{\mathsf{rot}}
\newcommand{\koucky}{Kouck\'y\xspace}

\crefname{property}{Property}{Properties}
\newtheorem{fact}[theorem]{Fact}
\theoremstyle{definition}
\newtheorem{problem}[theorem]{Problem}

\author{Gabriel Bathie}{DIENS, \'{E}cole normale sup\'{e}rieure de Paris, PSL Research University, France \and LaBRI, Université de Bordeaux, France}{gabriel.bathie@gmail.com}{https://orcid.org/0000-0003-2400-4914}{}
\author{Tomasz Kociumaka}{Max Planck Institute for Informatics, Saarland Informatics Campus, Germany}{tomasz.kociumaka@mpi-inf.mpg.de}{https://orcid.org/0000-0002-2477-1702}{}
\author{Tatiana Starikovskaya}{DIENS, \'{E}cole normale sup\'{e}rieure de Paris, PSL Research University, France}{tat.starikovskaya@gmail.com}{https://orcid.org/0000-0002-7193-9432}{}

\title{Small-Space Algorithms for the Online Language Distance Problem for Palindromes and Squares}
\date{\empty}

\authorrunning{G. Bathie, T. Kociumaka, and T. Starikovskaya}
\titlerunning{Online Language Distance Problem for Palindromes and Squares}
\Copyright{Gabriel Bathie, Tomasz Kociumaka, and Tatiana Starikovskaya}

\relatedversiondetails{Full version with details for the edit distance algorithms}{https://arxiv.org/abs/2309.14788} 


\funding{Partially funded by the grant ANR-20-CE48-0001 from the French National Research Agency.}

\hideLIPIcs
\nolinenumbers

\EventEditors{John Q. Open and Joan R. Access}
\EventNoEds{2}
\EventLongTitle{42nd Conference on Very Important Topics (CVIT 2016)}
\EventShortTitle{CVIT 2016}
\EventAcronym{CVIT}
\EventYear{2016}
\EventDate{December 24--27, 2016}
\EventLocation{Little Whinging, United Kingdom}
\EventLogo{}
\SeriesVolume{42}
\ArticleNo{23}

\begin{document}
\maketitle
\begin{abstract}
We study the online variant of the language distance problem for two classical formal languages, the language of palindromes and the language of squares, and for the two most fundamental distances, the Hamming distance and the edit  (Levenshtein) distance.
In this problem, defined for a fixed formal language $L$, we are given a string $T$ of length $n$, and the task is to compute the minimal distance to $L$ from \emph{every} prefix of~$T$.
We focus on the low-distance regime, where one must compute only the distances smaller than a given threshold $k$. In this work, our contribution is twofold:
\begin{enumerate}
\item First, we show \emph{streaming} algorithms, which access the input string $T$ only through a single left-to-right scan.
Both for palindromes and squares, our algorithms use $\Oh(k \polylog n)$ space and time per character in the Hamming-distance case
and $\Oh(k^2 \polylog n)$  space and time per character in the edit-distance case.
These algorithms are randomised by necessity, and they err with probability inverse-polynomial in $n$.
\item Second, we show \emph{deterministic read-only} online algorithms, which are also provided with read-only random access to the already processed characters of $T$.
 Both for palindromes and squares, our algorithms use $\Oh(k \polylog n)$ space and time per character in the Hamming-distance case and~$\Oh(k^4 \polylog n)$  space and amortised time per character in the edit-distance case.
\end{enumerate}
\end{abstract}

\section{Introduction}\label{sec:intro}
The \emph{language distance} problem is one of the most fundamental problems in formal language theory.
In this problem, the task is to compute the minimal distance between a given string~$S$ and any string belonging to a formal language $L$. Introduced in the early 1970s by Aho and Peterson~\cite{doi:10.1137/0201022}, the language distance problem has been studied extensively for regular languages under Hamming and edit distances~\cite{regular}, for general context-free languages, mainly focusing on the edit distance~\cite{ABW18,doi:10.1137/0201022,BGSW19,chi2022faster,10.1145/505241.505242,MYERS199585,RUZ79,7354391,8104067,Sat94}, and the Dyck language (the language of well-nested parentheses sequences) in particular~\cite{ABW18,BO16,BGSW19,chi2022faster,ApproxDyckED,DURR2023106358,FGK+22a,DBLP:conf/soda/0001S23,10.1007/978-3-642-22993-0_38,6979046,7354391,8104067}.

\subparagraph*{Our results.}
In this work, we study the complexity of the online and low-distance version of the language distance problem, where we are given a string $T$ of length $n$, and the task is to compute the minimal distance from \emph{every} prefix of $T$ to a formal language~$L$ (the distance and the language are specified in the problem definition).
We focus on the low-distance regime, i.e., we assume to be given a threshold parameter $k$ such that distances larger than $k$ do not need to be computed.
We consider the edit distance (defined as the minimum number of character insertions, deletions, and substitutions needed to transform one string into the other) and, as a preliminary step, the Hamming distance (allowing for substitutions only).
We study the problem for two classical languages: the language \PAL of all palindromes, where a palindrome is a string that is equal to its reversed copy, and the language \SQ of all squares, where a square is the concatenation of two copies of a string. These two languages are very similar yet very different in nature: \PAL is not regular but is context-free, whereas~\SQ is not even context-free. Formally, the problems we consider are defined as follows:

\begin{problem}{\kpalh{} (resp. \ksqh)}\\
\textbf{Input:} A string $T$ of length $n$ and a positive integer~$k$.\\
\textbf{Output:} For each $1 \le i \le n$, report $\min\{k+1,hd_i\}$, where $hd_i$
is the minimum Hamming distance between $T[1\dd i]$ and a string in \PAL (resp. in \SQ{}).
\end{problem}

\begin{problem}{\kpale{} (resp. \ksqe)}\\
\textbf{Input:} A string $T$ of length $n$ and a positive integer~$k$.\\
\textbf{Output:} For each $1 \le i \le n$, report $\min\{k+1,ed_i\}$, where $ed_i$
is the minimum edit distance between $T[1\dd i]$ and a string in \PAL (resp. in \SQ{}).
\end{problem}

\begin{table}[htbp]
	\centering
	\begin{tabular}{|l|l|l|l|l|}
	 \hline
	 Problem&Model& Time per character & Space complexity & Reference\\
	 \hline\hline
	 \kpalh & Streaming & $O(k\log^3 n)$ & $O(k \log n)$ &Thm~\ref{th:pal_ham}\\
	 \ksqh & Streaming & $\Ot(k)$ & $O(k \log^2 n)$ & Thm~\ref{th:square_ham}\\
	 \kpalh/SQ& Read-only & $O(k \log n)$ & $O(k \log n)$ & Thms~\ref{thm:ro-pal-ham} and~\ref{thm:ro-sq-ham}\\
	 \hline
	 \kpale/\SQ& Streaming & $\Ot(k^2)$ & $\Ot(k^2)$ & Thms~\ref{th:pal_edit} and~\ref{th:streaming_ksqe}\\
	 \kpale/SQ&Read-only& $\Ot(k^4)$\hfill (amortised) & $\Ot(k^4)$ & Thms~\ref{thm:ro-pal-e} and~\ref{thm:ro-sq-e}\\
	 \hline
	\end{tabular}
	\caption{Summary of the complexities of the algorithms introduced in this work.}
	\label{table:complexities}
\end{table}

Amir and Porat~\cite{amir2014approximate} showed that there is a randomised streaming algorithm that solves the \kpalh{} problem in $\Ot(k)$ space and $\Ot(k^2)$ time per input character.\footnote{Hereafter, $\Ot(\cdot)$ hides factors polynomial in $\log n$.}
We continue their line of research and show streaming algorithms for all four problems that use $\poly(k,\log n)$ time per character and $\poly(k,\log n)$ space.
While streaming algorithms are extremely efficient (in particular, the space complexities above account for \emph{all} the space used by the algorithms, including the space needed to store information about the input), it is important to note that they are randomised in nature, which means that they may produce incorrect results with a certain probability (inverse polynomial in the input size $n$).
Motivated by this, we also study the problems in the read-only model, where random access to the input is allowed (and not accounted for in the space usage).
In this model, we show \emph{deterministic} algorithms for the four problems that use $\poly(k,\log n)$ time per character and $\poly(k,\log n)$ space; see \cref{table:complexities} for a summary.
As a side result of independent interest, we develop the first $\poly(k, \log n)$ space read-only algorithms for computing $k$-mismatch and $k$-edit occurrences of a pattern in a text.

Due to the lack of space, descriptions of the algorithms for the Language edit distance problems (\kpale and \ksqe) are omitted from this version of the paper, but can be found in the full one. 

\subsection{Related work.}
\textbf{Offline model.} In the classical \emph{offline} model, the problem of finding \emph{all maximal substrings} that are within Hamming distance $k$ from \PAL can be solved in $O(nk)$ time as a simple application of the kangaroo jumps technique~\cite{galil1986improved}. For the edit distance, Porto and Barbosa~\cite{porto2002finding} showed an $O(n k^2)$ solution. For the \SQ language, the best known solutions take $O(n k \log k + \mathrm{output})$ time for the Hamming distance~\cite{KOLPAKOV2003135} and $O(nk\log^2 k + \mathrm{output})$ for the edit distance~\cite{LS93,SBT07,SOKOL2014103}.

\noindent\textbf{Online model.} The problems \kpalh and \kpale can be viewed as a generalization of the classical online palindrome recognition problem (see~\cite{Galil1976RealtimeAF} and references therein).

\noindent\textbf{Streaming algorithms for \PAL and \SQ.} Berebrink et al.~\cite{BEM+14} followed by Gawrychowski et al.~\cite{GawrychowskiMSU19} studied the question of computing the length of a maximal substring of a stream that belongs to \PAL. Merkurev and Shur~\cite{MS22} considered a similar question for the \SQ language.

\section{Preliminaries}
\label{sec:prelim}
We assume to be given an alphabet $\Sigma$, the elements of which, called \textit{characters}, can be stored in a single machine word of the RAM model.
For an integer $n\geq 0$, we denote the set of all length-$n$ strings by $\Sigma^n$, and we set $\Sigma^{\le n} = \bigcup_{m=0}^n \Sigma^m$ as well as $\Sigma^* = \bigcup_{n= 0}^\infty \Sigma^n$. The empty string is denoted by $\eps$.

For two strings $S,T \in\Sigma^*$, we use $ST$ or $S\cdot T$ indifferently to denote their concatenation.
For an integer $m \ge 0$, the string obtained by concatenating $S$ to itself $m$ times is denoted by~$S^m$; note that $S^0=\eps$.
A string $S$ is a \emph{square} if there exists a string $T$ such that $S = T^2$.

For a string $T \in\Sigma^n$ and an index $i\in [1\dd n]$,\footnote{For integers $i,j\in \mathbb{Z}$, denote $[i\dd j] =
\{k \in \mathbb{Z} : i \le k \le j\}$, $[i
\dd j)=\{k \in \mathbb{Z} : i \le k <
j\}$, and $(i\dd j]={\{k \in \mathbb{Z}: i
< k \le j\}}$.} the $i$th character of $T$ is denoted by $T[i]$.
We use $|T| = n$ to denote the length of $T$.
For indices $1 \le i, j\le n$,
$T[i\dd j]$ denotes the substring $T[i] T[{i+1}]\cdots T[j]$ of~$T$ if $i\le j$
and the empty string otherwise.  
When $i=1$ or $j=n$, we omit these indices, i.e., we write $T[\dd j] = T[1\dd j]$
and $T[i\dd ] = T[i\dd n]$.
We extend the above notation with $T[i\dd j) = T[i\dd j-1]$
and $T(i\dd j] = T[i+1\dd j]$.
We say that a string $P$ is a \emph{prefix} of $T$ if there exists $j\in [1\dd n]$
such that $P = T[\dd j]$, and a \emph{suffix} of $T$ if there exists $i\in [1\dd n]$
such that $P = T[i\dd ]$.
We use $T^R$ to denote the reverse of $T$,
that is~$T^R = T[n]T[n-1]\cdots T[1]$.
A string $T$ is a \emph{palindrome} if $T^R = T$.

We define the \emph{forward cyclic rotation} $\rot(T) = T[2]\cdots T[n]T[1]$. In general, a cyclic rotation $\rot^s(T)$ with shift $s \in \mathbb{Z}$ is obtained by iterating $\rot$ or the inverse operation $\rot^{-1}$. A non-empty string $T\in \Sigma^n$ is \emph{primitive} if it is distinct from its non-trivial rotations, i.e., if~$T = \rot^s(T)$ holds only when $n$ divides~$s$.

Given two strings  $U,V$ and two indices $i \in [1\dd |U|]$ and $j \in [1\dd  |V|]$,
the \emph{longest common prefix} (\lcp) of $U[i\dd ]$ and $V[j\dd ]$, denoted $\lcp(U[i\dd ],V[j\dd ])$,
is the length of the longest string that is a prefix of both $U[i\dd ]$ and~$V[j\dd ]$.

Given two non-empty strings $U, Q$ and an operator $F$ defined over pairs of strings,
we use the notation $F(U, Q^\infty)$ for the application of $F$ to $U$ and the prefix of $Q^\infty = QQ\cdots$
that has the same length as $U$, i.e., $F(U, Q^\infty) = F(U, Q^m[\dd |U|])$,
where $m$ is any integer such that $|Q^m| \ge |U|$.
We define $F(Q^\infty, U)$ symmetrically.

\subsection{Hamming distance, palindromes, and squares}
The Hamming distance between two strings $S, T$ (denoted $\hd{S}{T}$) is defined to be equal
to infinity if $S$ and $T$ have different lengths,
and otherwise to the number of positions where the two strings differ (mismatches).
We define the \emph{mismatch information} between two length-$n$ strings $S$ and $T$,
$\MI(S,T)$  as the set $\{(i, S[i], T[i]) : i\in [1\dd n]\text{ and } S[i] \neq T[i]\}$.
For two strings $P, T$, a position $i\in [|P|\dd |T|]$ of $T$ is a \textit{$k$-mismatch occurrence}
of $P$ in~$T$ if~$\hd{T(i-|P|\dd i]}{P} \le k$. For an integer $k$, we denote $\hdk{X}{Y}{k}= \hd{X}{Y}$ if $\hd{X}{Y}\le k$ and $\infty$ otherwise.

Due to the self-similarity of palindromes and squares,
the Hamming distance from a string $U$ to \PAL and \SQ can be measured in terms of the self-similarity of $U$.

\begin{restatable}{property}{propkpalham}\label{prop:kpalham}
    Each string $U \in \Sigma^m$ satisfies $\hd{U}{\PAL}= \hd{U[ \dd \lfloor{m/2}\rfloor]}{U(\lceil{m/2}\rceil\dd]^R} = \frac12\hd{U}{U^R}$.
\end{restatable}
\begin{proof}
    Denote $U_1=U[ \dd \lfloor{m/2}\rfloor]$ and $U_2=U(\lceil{m/2}\rceil\dd ]$.
    For the second equality, we have $\hd{U}{U^R} = \hd{U_1}{U_2^R} + \hd{U_2}{U_1^R} = 2\cdot \hd{U_1}{U_2^R}$.
   
    The first equality is equivalent to $\hd{U_1}{U_2^R} = \hd{U}{\PAL}$. As the Hamming distance between $U$ and the palindrome $U_2^RU_2$
    (or $U_2^RaU_2$ if $m$ is odd) is $\hd{U_1}{U_2^R}$, we have~$\hd{U_1}{U_2^R} \geq \hd{U}{\PAL}$.

    Conversely, let $V$ be a palindrome such that $\hd{U}{V} = \hd{U}{\PAL}$.
    We decompose similarly~$V$ into $V_1V_1^R$ (or $V_1bV_1^R$
    for odd $m$)   and obtain $\hd{U}{V} \ge \hd{U_1}{V_1} + \hd{U_2}{V_1^R}$.
    Using the fact that $\hd{U_2}{V_1^R} = \hd{U_2^R}{V_1}$ and applying the triangle inequality,
    we get~$\hd{U_1}{U_2^R}\le \hd{U}{\PAL}$.
\end{proof}

\begin{restatable}{property}{propksqham}\label{prop:k-sq-ham}
    Each string $U\in \Sigma^m$ satisfies $\hd{U}{\SQ}=\hd{U[\dd m/2]}{U(m/2\dd ]}$ if $m$ is even
    and $\hd{U}{\SQ} = \infty$ if $m$ is odd.
\end{restatable}
\begin{proof}
    Every square has even length; hence, if $m$ is odd, the distance between $U$
    and \SQ is infinite. In what follows, we assume that $m = 2i$ for some $i\in\NN$.
    Let $U_1 = U[\dd i]$ and $U_2 = U(i\dd ]$.
    By modifying the copy of $U_1$ in $U$ into $U_2$, we obtain a square $U_2U_2$; hence, $\hd{U}{\SQ} \le \hd{U_1}{U_2}$.

    For the converse inequality, let $V^2$ be a square such that $\hd{U}{\SQ} = \hd{U}{V^2}$.
    We have $|V| = |U_1| = |U_2|$; hence, $\hd{U}{V^2} = \hd{U_1}{V} + \hd{V}{U_2}$.
    Applying the triangle inequality,
    we obtain $\hd{U}{\SQ} = \hd{U}{V^2} \ge \hd{U_1}{U_2}$.
\end{proof}

\subsection{Models of computation}
In this work, we focus on two by now classical models of computation: streaming and read-only random access.
In the streaming model, we assume that the input string $T$ arrives as a stream, one character at a time. For each prefix $T[1\dd i]$, we must report the distance to \PAL or \SQ as soon as we receive $T[i]$. We account for all the space used, including the space needed to store any information about $T$.
In contrast, in the read-only model, we do not account for the space occupied by the input string. We assume that $T$ is read from the left to the right. After having read $T[1\dd i]$, we assume to have constant-time read-only random access to the first~$i$ characters of $T$. Similar to the streaming model, the distance from $T[1\dd i]$ to \PAL or \SQ must be reported as soon as we read $T[i]$. 

\section{Warm-up: Streaming algorithms for the LHD problems}\label{sec:ham}
In this section, we present streaming algorithms for \kpalh and \ksqh.
Our solutions use the Hamming distance sketches introduced by Clifford, Kociumaka, and Porat~\cite{clifford2019streaming}
to solve the streaming~$k$-mismatch problem.

\begin{fact}\label{fact:ham_sketch}
    There exists a function $\sk_k$ (parameterized by a constant $c>1$, integers $n \ge k \ge 1$, and a seed of $\Oh(\log n)$ random bits)
    that assigns an $\Oh(k\log n)$-bit sketch to each string in~$\Sigma^{\le n}$. Moreover:
    \begin{enumerate}
        \item There is an $\Oh(k\log^2 n)$-time \emph{encoding} algorithm that, given $U\in \Sigma^{\le k}$, builds $\skk{U}$.
        \item There is an $\Oh(k\log n)$-time algorithm that, given any two among $\skk{U}, \skk{V}$, or $\skk{UV}$, computes the third one (provided that $|UV|\le n$).
        \item There is an $O(k \log^3 n)$-time \emph{decoding} algorithm that, given $\skk{U}$ and $\skk{V}$, computes $\MI(U,V)$ if $\hd{U}{V}\le k$. The error probability is $O(n^{-c})$.
    \end{enumerate}
\end{fact}

\subsection{A streaming algorithm for \texorpdfstring{\kpalh}{k-LHD-PAL}}
We first show that the sketches described in~\cref{fact:ham_sketch} give a simple algorithm improving upon the result of Amir and Porat~\cite{amir2014approximate} and achieving the time complexity of $\Ot(k)$ per letter.

\begin{restatable}{theorem}{thmpalham}\label{th:pal_ham}
    There is a randomised streaming algorithm that solves the \kpalh{} problem for a string $T \in \Sigma^n$ using $\Oh(k\log n)$ bits of space and $O(k\log^3 n)$ time per character. The algorithm errs with probability inverse-polynomial in $n$.
\end{restatable}

Using Property~\ref{prop:kpalham}, we can reduce the \kpalh{} problem
to that of computing the threshold Hamming distance between the current prefix of the input string and its reverse.
The algorithm maintains the sketches $\sk_{2k}(T[\dd i])$ and $\sk_{2k}(T[\dd i]^R)$. When it receives~$T[i]$, it constructs $\sk_{2k}(T[i])$,
updates both $\sk_{2k}(T[\dd i])$ and $\sk_{2k}(T[\dd i]^R)$, and computes $d = \hdk{T[\dd i]}{T[\dd i]^R}{2k}$ (in $O(k \log^3 n)$ total time by \cref{fact:ham_sketch}). \Cref{prop:kpalham} implies $\hdk{T[\dd i]}{\PAL}{k} = d/2$.
The error probability of the algorithm follows from the error probability for the decoding algorithm for Hamming distance sketches.

The algorithm uses $\Oh(k \log n)$ bits, which is nearly optimal:
Indeed, by Property~\ref{prop:kpalham}, if~$U = VW$, with $|V| = |W|$,
then $\hd{U}{U^R} = 2\cdot \hd{V}{W^R}$.
Therefore, using a standard reduction from streaming algorithms
to one-way communication complexity protocols,
we obtain a lower bound of $\Omega(k)$ bits for the space complexity
of streaming algorithms for the \kpalh problem
from the $\Omega(k)$ bits lower bound for the communication complexity of the Hamming distance~\cite{huang2006communication}.

\subsection{A streaming algorithm for \texorpdfstring{\ksqh}{k-LHD-SQ}}\label{sec:ksqh}
In this section, we show the following theorem:

\begin{restatable}{theorem}{thmsquareham}\label{th:square_ham}
    There is a randomised streaming algorithm that solves the \ksqh{} problem for a
    string $T \in \Sigma^n$ using $\Oh(k\log^2 n)$ bits of space and $\Ot(k)$ time per character.
    The algorithm errs with probability inverse-polynomial in $n$.
\end{restatable}


Property~\ref{prop:k-sq-ham} allows us to derive $\hdk{T[\dd 2i]}{\SQ}{k}$ from the sketches $\skk{T[\dd i]}$
and~$\skk{T[\dd 2i]}$: we can combine them to obtain $\skk{T(i\dd 2i]}$,
and a distance computation on $\skk{T[\dd i]}$ and $\skk{T(i\dd 2i]}$
returns $\hdk{T[\dd i]}{T(i\dd 2i]}{k}=\hdk{T[\dd 2i]}{\SQ}{k}$.

Naively applying this procedure requires storing the sketch $\skk{T[\dd i]}$
until the algorithm has read $T[\dd 2i]$, that is, storing $\Theta(n)$ sketches at the same time.
To reduce the number of sketches stored, we use a filtering procedure based on the following observation:
\begin{observation}\label{obs:ksqh}
If $\hd{T[\dd 2i]}{\SQ}\le k$ and $\ell\in [1\dd i]$, then $i+\ell$ is a~$k$-mismatch occurrence of $T[\dd \ell]$,
that is, $\hd{T[\dd \ell]}{T(i\dd i+\ell]}\le k$.
\end{observation}
\begin{example}
For $k = 1, \ell = 2$, and $i = 3$, the word $T[\dd 6] = \mathtt{abc\underline{ac}c}$ is a $1$-mismatch square
(by \cref{prop:k-sq-ham}) and the fragment $T(3\dd 5] = \mathtt{ac}$ is a 1-mismatch occurrence of the prefix $T[\dd 2]= \mathtt{ab}$.
\end{example}

\Cref{obs:ksqh} motivates our filtering procedure:
if we choose some prefix $P = T[\dd \ell]$ of the string,
we only need to store every $i \ge \ell$ such that
$i+\ell$ is a $k$-mismatch occurrence of $P$. Clifford, Kociumaka and Porat~\cite{clifford2019streaming} showed a data structure $\Ss$
that exploits the structure of such occurrences and stores them using $O(k\log^2 n)$ bits of space while allowing reporting the occurrence at position $i+\ell$ when $T[i+\ell+\Delta]$ is pushed into $\Ss$
-- we say that $\Ss$ reports the $k$-mismatch occurrences of $P$ in $T$ with \emph{a fixed delay} $\Delta$~\cite{clifford2019streaming}.
Our algorithm needs to receive the occurrence at position $i+\ell$ when $T[2i]$ is pushed into the stream,
i.e. we require $\Ss$ to report occurrences with a \textit{non-decreasing} delay. In~\cref{sec:incr-delay} we present a modification of the data structure~\cite{clifford2019streaming} to allow non-decreasing delays, and in \cref{sec:alg-k-sq-h} we explain how we use it to implement a space-efficient streaming algorithm for \ksqh.

\subsubsection{Reporting $k$-mismatch occurrences with nondecreasing delay.}\label{sec:incr-delay}

The algorithm of Clifford, Kociumaka, and Porat~\cite{clifford2019streaming}
reports additional information along with the positions of the $k$-mismatch occurrences:
specifically, it produces the \emph{stream of $k$-mismatch occurrences of $P$ in $T$},
defined as follows.
\newcommand{\kstream}{S_P^k}
\newcommand{\kstreamj}{S_{P_j}^k}
\begin{definition}[{\cite[Definition 3.2]{clifford2019streaming}}]
    The stream of $k$-mismatch occurrences of a pattern~$P$ in a text $T$ is a sequence $\kstream$ such that $\kstream[i] =  (i,\;\MI(T(i-|P|\dd i],P),\;\skk{T[\dd i-|P|]})$ if $\hd{P}{T(i-|P|\dd i]} \le k$ and $\kstream[i]=\bot$ otherwise.
\end{definition}
As explained next, the algorithm of~\cite{clifford2019streaming} can report the $k$-mismatch occurrences
with a prescribed delay.

\newcommand{\Dd}{\mathcal{D}}
\newcommand{\Bb}{\mathcal{B}}
\begin{corollary}[of {\cite{clifford2019streaming}}]\label{fact:encode-ham-occ}
    There is a streaming algorithm that, given a pattern $P$ followed by a text $T \in \Sigma^n$,
    reports the $k$-mismatch occurrences of $P$ in $T$
    using $O(k\log^2 n)$ bits of space and $O(\sqrt{k\log^3 n}+\log^4 n)$ time per character.
    The algorithm can report each occurrence $i$ with no delay (that is, upon receiving $T[i]$)
    or with any prescribed delay $\Delta = \Theta(|P|)$ (that is, upon receiving $T[i+\Delta]$).
    For each reported occurrence~$i$, the underlying tuple $\kstream[i]$ can be provided on request in $\Oh(k\log^2 n)$ time.
\end{corollary}
\begin{proof}
    If no delay is required, we use~\cite[Theorem 1.2]{clifford2019streaming}, which reports $k$-mismatch occurrences of $P$ in $T$
    and, upon request, provides the mismatch information $\MI(T(i-|P| \dd i],P)$;
    this algorithm uses $\Oh(k\log^2 n)$ bits of space and takes $O(\sqrt{k\log^3 n}+\log^4 n)$ time per character.
    We also use \cite[Fact~4.4]{clifford2019streaming} to maintain the sketch $\skk{T[\dd i]}$ (reported on request);
    this algorithm uses $\Oh(k\log n)$ bits of space and takes $O(\log^2 n)$ time per character.

    Whenever requested to provide $\kstream[i]$ for some $k$-mismatch occurrence $i$ of $P$ in $T$,
    we retrieve the mismatch information $\MI(T(i-|P|\dd i],P)$ (in $\Oh(k)$ time)
    and the sketch $\skk{T[\dd i]}$ (in $\Oh(k \log^2 n)$ time).
    Combining $\skk{P}$ with $\MI(T(i-|P|\dd i],P)$, we build $\skk{T(i-|P|\dd i]}$ (using \cite[Lemma 6.4]{clifford2019streaming} in $\Oh(k\log^2 n)$ time) and then derive $\skk{T[\dd i-|P|]}$ using \cref{fact:ham_sketch} (in $\Oh(k\log n)$ time).
    Overall, processing the request takes $\Oh(k\log^2 n)$ time and $\Oh(k\log^2 n)$ bits of space.

    If a delay $\Delta = \Theta(|P|)$ is required, our approach depends on whether there exists $p\in [1\dd k]$ such that $\hd{P[\dd |P|-p]}{P(p\dd |P|]}\le 2k$
    (such $p$ is called a $2k$-period in \cite{clifford2019streaming}).
    This property is tested using a streaming algorithm of \cite[Lemma 4.3]{clifford2019streaming}, which takes $\Oh(k\log n)$ bits of space,
    $\Oh(\sqrt{k\log n})$ time per character of $P$, and requires $\Oh(k\sqrt{k\log n})$-time post-processing (performed while reading $T[\dd k]$).
    If $P$ satisfies this condition, then we just use \cite[Theorem 4.2]{clifford2019streaming}, whose statement matches that of \cref{fact:encode-ham-occ}.

    Otherwise, \cite[Observation 4.1]{clifford2019streaming} shows that $P$ has at most one $k$-mismatch occurrence among any $k$ consecutive positions in $T$.
    In that case, we use the aforementioned approach to produce the stream $\kstream$ with no delay and the buffer of \cite[Proposition 3.3]{clifford2019streaming} to delay the stream by $\Delta$ characters.
    The buffering algorithm takes $\Oh(k\log^2 n)$ bits of space and processes each character $T[i]$ in $\Oh(k\log^2 n + \log^3 n)$ time (if $P$ has $k$-mismatch occurrences at positions~$i$ or $i-\Delta$) or $O(\sqrt{k\log n}+\log^3 n)$ time (otherwise).
    Since the former case holds for at most two out of every $k$ consecutive positions, we can achieve  $O(\sqrt{k\log^3 n}+\log^4 n)$ worst-case time per character by decreasing the delay to $\Delta-k$ and buffering up to $k$ characters of $T$ and up to $k$ elements of $\kstream$.
    While the algorithm processes $T[i+\Delta]$, the latter buffer already contains $\kstream[i]$, but $\Oh(k)$ time is still needed to output this value (if $\kstream[i]\ne \bot$).
\end{proof}

The algorithm of \cref{fact:encode-ham-occ} has a fixed delay $\Delta$, i.e., it outputs $\kstream[i]$
upon receiving $T[i+\Delta]$.
Our application requires a variable delay: we need to access $\kstream[i+|P|]$ upon reading~$T[2i]$,
that is, with a delay of $i-|P|$.
We present a black-box construction that extends the data structure of \cref{fact:encode-ham-occ}
to support non-decreasing delays $\Delta_i$, $i \in [1\dd d]$. Naively, one could use the algorithm $\Alg$ of \cref{fact:encode-ham-occ}
with a fixed delay $\Delta_1$ and buffer the input characters so that $\Alg$ receives $T[i+\Delta_1]$ only when we actually process $T[i+\Delta_i]$.
Unfortunately, this requires storing $T[i+\Delta_1\dd i+\Delta_i)$, which could take too much space.
Thus, we feed $\Alg$ with $T[1\dd \Delta_1]$ followed by blank characters $\bot$ (issued at appropriate time steps without the necessity of buffering input characters) so that $\Alg$ reports $k$-mismatch occurrences $i\in [1\dd \Delta_1]$ with prescribed delays.
Then, we use another instance of the algorithm of \cref{fact:encode-ham-occ},
with a fixed delay $\Delta_{1+\Delta_1}$, to output $k$-mismatch occurrences $i\in (\Delta_1\dd \Delta_1 + \Delta_{1+\Delta_1}]$;
we continue this way until the whole interval $[1\dd d]$ is covered.  We formalise this idea in the following lemma.

\begin{restatable}{lemma}{variabledelay}\label{lemma:variable-delay}
    Let $\Delta_1 \le \Delta_2 \le \cdots \le \Delta_d$ be a non-decreasing sequence of $d=\Oh(|P|)$ integers
    $\Delta_i = \Theta(|P|)$, represented by an oracle that reports each element $\Delta_i$ in constant time.

    There is a streaming algorithm that, given a pattern $P$ followed by a text $T$,
    reports the $k$-mismatch occurrences of $P$ in $T$
    using $O(k\log^2 n)$ bits of space and $O(\sqrt{k\log^3 n}+\log^4 n)$ time per character.
    The algorithm reports each occurrence $i\in [1\dd d]$ with delay $\Delta_i$, that is, upon receiving $T[i+\Delta_i]$.
    For each reported occurrence $i\in [1\dd d]$, the underlying tuple $\kstream[i]$ can be provided on request in $\Oh(k\log n)$ time.
\end{restatable}
\begin{proof}
    We use multiple instances $\Alg_1,\ldots,\Alg_t$ of the algorithm of \cref{fact:encode-ham-occ}.
    We define a sequence $(s_r)_{r=0}^{t}$ so that $\Alg_r$ works with a fixed delay $\Delta_{s_{r-1}}$,
    it is given $T[1\dd s_{r})\cdot \bot^{\Delta_{s_{r-1}}}$, and it reports $k$-mismatch occurrences $i\in [s_{r-1}\dd s_{r})$. Specifically, we set $s_0 = 1$ and $s_{r}=s_{r-1}+\Delta_{s_{r-1}}$,
    with $t$ chosen as the smallest integer such that $s_{t} > d$.
    Note that $s_{r}-s_{r-1} = \Delta_{s_{r-1}} \ge \Delta_1$ implies $t\le 1 + \frac{d}{\Delta_1}= \Oh(1)$.

    We assign three different roles to the algorithms $\Alg_1,\ldots,\Alg_r$:
    \emph{passive}, \emph{active}, and \emph{inactive}.
    While we process $T[j]$, the algorithm $\Alg_r$ is passive if $j < s_{r}$,
    active if $j\in [s_r\dd s_{r+1})$, and inactive if $j\ge s_{r+1}$.
    Our invariant is that, once we process $T[j]$, each passive algorithm $\Alg_r$ has already received $T[1\dd j]$,
    the unique active algorithm $\Alg_r$ has already received $T[1\dd s_{r})\cdot \bot^{1+i-s_{r-1}}$,
    where $i$ is the largest integer such that $i+\Delta_i \le j$,
    and each inactive algorithm~$\Alg_r$ has already received its entire input, that is,  $T[1\dd s_{r})\cdot \bot^{\Delta_{s_{r-1}}}$.

    Upon receiving $T[j]$, we simply forward $T[j]$ to all passive algorithms.
    Moreover, if $j = i+\Delta_i$ for some $i\in [1\dd d]$, we feed the active algorithm with $\bot$
    so that it checks whether $i$ is a $k$-mismatch occurrence of $P$ in $T$ and, upon request, outputs $\kstream[i]$.

    Let us argue that this approach is correct from the perspective of a fixed algorithm $\Alg_r$. As we process $T[1\dd s_r)$,
    the algorithm is passive, and it is fed with subsequent characters of $T$.
    For $j=s_r-1$, the position $i=s_{r-1}-1$ is the maximum one such that $i+\Delta_i \le j$.
    Consequently, the input $T[1\dd s_r)$ already satisfies the invariant for passive algorithms.
    For subsequent iterations $j\in [s_r\dd s_{r+1})$, as $\Alg_r$ is active, it receives $\bot$ whenever $i$ increases,
    so its input stays equal to $T[1\dd s_{r})\cdot \bot^{1+i-s_{r-1}}$.
    The length of this string is $s_r+i-s_{r-1}=i+\Delta_{s_{r-1}}$, so the algorithm indeed checks whether $i$ is a $k$-mismatch occurrence of $P$ in $T$ at each such iteration (recall that its fixed delay is $\Delta_{s_{r-1}}$), and it satisfies the invariant for active algorithms.
    Once we reach $j=s_{r+1}-1$, we have $i=s_{r}-1=s_{r-1}+\Delta_{s_{r-1}}-1$, so the input becomes $T[1\dd s_{r})\cdot \bot^{\Delta_{s_{r-1}}}$,
    and it already satisfies the invariant for inactive algorithms. The state of inactive algorithms does not change, so this invariant remains satisfied as $\Alg_r$ stays inactive indefinitely.

    The time and space complexity analysis follows from the fact that $t = O(1)$.
\end{proof}

\subsubsection{Algorithm}\label{sec:alg-k-sq-h}

We now show how to use the data structure of \cref{lemma:variable-delay} to implement our filtering procedure using low space.
For each $j\in [1\dd \lfloor \log n \rfloor]$, let $P_j$ denote the prefix of the text of length $\ell_j = 2^j$, i.e., $P_j = T[\dd 2^j]$. We search for $k$-mismatch occurrences of $P_j$ in~$T_j =T(3\ell_j/2 \dd 4\ell_j]$. As argued below, this allows filtering positions in $(3\ell_j\dd  6\ell_j]$. Additionally, our choice of $(\ell_j)_{j}$ ensures that we do not miss any $k$-mismatch square when running our search for every $P_j$ in parallel.

\begin{restatable}{claim}{streamingksqocc}\label{claim:streaming-k-sq-occ}
    For each $j\in [1\dd \lfloor \log n \rfloor]$, let $\occ_j$ be the set of $k$-mismatch occurrences of $P_j$ in $T_j = T(3\ell_j/2\dd 4\ell_j]$.
    If $\hd{T[\dd 2i]}{\SQ}\le k$ and $2i\in[3\ell_j\dd 6\ell_j)$, then $p = i - \ell_j/2 \in \occ_{j}$.
\end{restatable}
\begin{claimproof}
    Since $\ell_j \le i$, \cref{obs:ksqh} implies that $i+\ell_j$ is a $k$-mismatch occurrence of $P_j$ in $T$.
    Moreover, when $2i\in[3\ell_j\dd 6\ell_j)$, we have $3\ell_j/2 \le i \le 3\ell_j$;
    therefore, that $k$-mismatch occurrence of $P_j$ is fully contained within $T_j$,
    and it ends at positions $i+\ell_j - 3\ell_j/2 = i -\ell_j/2$ of $T_j$.
\end{claimproof}

In what follows, we use $p$ to denote indices in $T_j$, whereas $i$ denotes indices in the original text $T$.
As $T_j = T(3\ell_j/2\dd 4\ell_j]$, the correspondence is given by $i = p + 3\ell_j/2$.
In other words, we only need to compute $\hdk{T[\dd 2i]}{\SQ}{k}$ when $i - \ell_j/2\in\occ_j$.
As noted in Property~\ref{prop:k-sq-ham},
it suffices to know the sketches $\skk{T(i\dd 2i]}$ and $\skk{T[\dd i]}$.
We store $\skk{P_j}=\skk{T[\dd \ell_j]}$ as well as $s_j=\skk{T[\dd 3\ell_j/2]}$ and maintain $\skk{T[\dd 2i]}$ in a rolling manner as we receive the characters of the text.

We use the algorithm of \cref{lemma:variable-delay}, asking for $k$-mismatch occurrences of $P_j$ in $T_j$,
to report $\skk{T_j[\dd i-\ell_j]} = \skk{T(\ell_j\dd i]}$ for every $i \in \occ_j$.
The delay sequence is specified as $\Delta_p = p-\ell_j/2$ for $p\in [\ell_j\dd 5\ell_j/2)$
so that the conditions of \cref{lemma:variable-delay} are satisfied.
(For $p < \ell_j$, we can assume $\Delta_p = \Delta_{\ell_j}=\ell_j/2$; anyway, there cannot be a $k$-mismatch occurrence of $P_j$ before position $\ell_j$.)
This way, for every $i\in [3\ell_j/2\dd 3\ell_j)$,
we receive $\kstreamj[i+\ell_j]$ (which corresponds to a potential $k$-mismatch occurrence starting at position $i+1$)
while processing $T_j[p+\Delta_{p}]$ for $p = i + \ell_j - 3\ell_j/2 = i - \ell_j/2$.
As $\Delta_p = p-\ell_j/2$, this corresponds to position $p' = 2p - \ell_j/2$ in $T_j$, or position
$i' = 2p +\ell_j = 2i$ in $T$, i.e., this happens precisely as we are processing~$T[2i]$.
See \cref{fig:filter-ksqh} for an illustration of the above.
If~$\kstreamj[i+\ell_j]$ is blank, we move on to the next position.
Otherwise, we retrieve the sketch $\skk{T_j[\dd i]} = \skk{T(3\ell_j/2\dd i]}$,
combine it with $s_j=\skk{T[\dd 3\ell_j/2]}$ and $\skk{T[\dd 2i]}$ to obtain $\skk{T[\dd i]}$ and $\skk{T(i\dd 2i]}$,
and use the latter two sketches to compute
$\hdk{T[\dd i]}{T(i\dd 2i]}{k}$, which is equal to $\hdk{T[\dd 2i]}{\SQ}{k}$ by \cref{prop:k-sq-ham}.

\begin{figure}[htbp]
    \tikzmath{
    \plen = 5; 
    \slen = 9;
    \tlen = 21.5;
    \posheight = -.4;
    \tjheight = 2;
    }
    \centering
    \begin{tikzpicture}[scale=0.6, every node/.style={scale=.9}]
    	\node (v) at (\tlen + .5, 0.5) {$T$};
        \draw (0, 0) rectangle ++(\tlen, 1);
    	
    	\node (3lj2) at (3*\plen/2, \posheight) {$3\ell_j/2$};
        \draw[dashed] (3*\plen/2, 0) -- ++(0, \tjheight);
    	
    	\node (4lj) at (4*\plen, \posheight) {$4\ell_j$};
        \draw[dashed] (4*\plen, 0) -- ++(0, \tjheight);

    	\node (i) at (\slen, \posheight) {$i$};

    	\node (ilj) at (\slen+\plen, \posheight) {$i+\ell_j$};

    	\node (ii) at (2*\slen, \posheight) {$2i$};
        \draw[dashed] (2*\slen, 0) -- ++(0, \tjheight+1);

        \draw (0, 0) rectangle ++(\plen, 1) node [midway] {$P_j$};
        \draw (\slen, 0) rectangle ++(\plen, 1) node [midway] {$P'$};

        \draw (3*\plen/2, \tjheight) rectangle ++(5*\plen/2, 1);
        \node (tj) at (4*\plen+.7, \tjheight+.5) {$T_j$};

    	\node (z) at (3*\plen/2, \tjheight+1.4) {$0$};
    	
    	\node (p) at (\slen+\plen, \tjheight+3) {$p = i-\ell_j/2$};
        \draw[->] (p) -> (\slen+\plen, \tjheight+1.2);
        \draw[dashed] (\slen+\plen, \tjheight+1) -- ++(0, -\tjheight);
        \draw[|<->|] (\slen+\plen, \tjheight+1.5) -- (2*\slen, \tjheight+1.5) node[above, midway] {$\Delta_p = p-\ell_j/2$};

    \end{tikzpicture}
	\caption{
		Illustration of our filtering procedure.
		Here, $P'$ is a $k$-mismatch occurrence of $P_j$ at position $i + \ell_j$ in $T$ and position $p = i -\ell_j/2$ in $T_j$,
		reported with delay $\Delta_p = p-\ell_j/2$ in $T_j$, hence it arrives at time $2i$ in $T$.
	}
	\label{fig:filter-ksqh}
\end{figure}
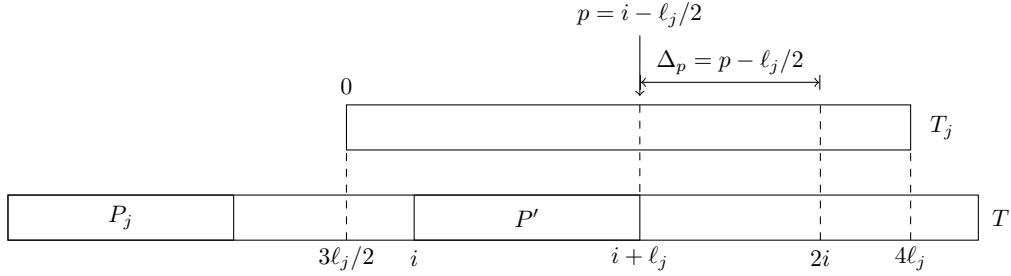

We proceed with the complexity analysis of our algorithm.
The $k$-mismatch pattern matching algorithm of \cref{lemma:variable-delay} uses $O(k \log^2 n)$ bits of space and $\Ot(k)$ time per character,
and we maintain $\Oh(\log n)$ instances of this algorithm. However, since all the patterns $P_j$ are prefixes of $T$,
the instances can share the pattern processing phase. Moreover, since any position is contained in at most three fragments $T[\ell_j\dd 6\ell_j)$ (each such fragment follows $P_j$ and contains $T_j$), at most three instances contribute to the time and space complexity at any given moment.
Thus, the entire algorithm uses $\Oh(k\log^2 n)$ bits of space and $\Ot(k)$ time per character, which completes the proof of~\cref{th:square_ham}.

Our streaming algorithm for \ksqe{} (\cref{th:streaming_ksqe}) relies on the streaming algorithm for \ksqh{}.
It requires testing $\hd{T[\dd 2i]}{\SQ}\le k$ only for selected positions $i$, and thus it benefits from the following variant of \cref{th:square_ham}:

\begin{restatable}{proposition}{hdsqmi}\label{rm:hd-sq-mi}
    There is a randomised streaming algorithm that, given a string $T \in \Sigma^n$, upon receiving~$T[2i]$,
    can be requested to test whether $\hd{T[\dd 2i]}{\SQ}\le k$ and, if so, report the mismatch information between $T[\dd 2i]$ and a closest square.
    The algorithm uses $\Oh(k\log^2 n)$ bits of space and processes each character in $\Ot(\sqrt{k})$ or $\Ot(k)$ time, depending on whether the request has been issued at that character.
\end{restatable}
\begin{proof}
    We follow the algorithm above with minor modifications. First, instead of maintaining $\skk{T[\dd 2i]}$ explicitly,
    we apply \cite[Fact 4.4]{clifford2019streaming}, which uses $\Oh(k\log n)$ bits of space, takes $\Oh(\log^2 n)$ time per character,
    and reports $\skk{T[\dd 2i]}$ on demand in $\Oh(k\log^2 n)$ time.

    To process a request concerning position $2i$, we retrieve $\skk{T[\dd 2i]}$ and ask the pattern-matching algorithm of \cref{lemma:variable-delay} to output $\kstreamj[i]$ (normally, the algorithm only reports whether $i$ is a $k$-mismatch occurrence of $P_j$ in $T_j$). In this case, we build $\skk{T[\dd i]}$ and $\skk{T(i\dd 2i]}$ as in algorithm above. The decoding algorithm not only results in~$\hdk{T[\dd i]}{T(i\dd 2i]}{k}=\hdk{T[\dd 2i]}{\SQ}{k}$ but, if $\hd{T[\dd 2i]}{\SQ}\le k$, also the underlying mismatch information.

    The space complexity of the modified algorithm is still $\Oh(k\log^2 n)$ bits.
    The running time is $\Ot(\sqrt{k})$ if we do not ask the algorithm to test $\hd{T[\dd 2i]}{\SQ}\le k$
    and $\Ot(k)$ if we do.
\end{proof}

\section{Deterministic read-only algorithms for the LHD problems}\label{sec:ro-pal}
In this section, we present deterministic read-only algorithms for \kpalh and \ksqh.
We start by recalling structural results for $k$-mismatch occurrences used by the algorithms.

\subsection{Structure of \texorpdfstring{\boldmath $k$}{k}-mismatch occurrences}
\begin{definition}[{\cite{charalampopoulos2020faster}}]\label{def:k-mism-periodic}
    A string $U$ is \textit{$d$-mismatch periodic} if there exists a \textit{primitive} string~$Q$
    such that $|Q| \leq |U|/128d$ and $\hd{U}{Q^\infty}\le 2d$.
    Such a string $Q$ is called the \emph{$d$-mismatch period} of $U$.
\end{definition}
The condition $|Q| \leq |U|/128d$ implies that $Q$ is equal to some substring of $U$; hence, given the starting and ending positions of $Q$ in $U$ and random access to $U$, we can simulate random access to $Q$.

\begin{claim}[{From~\cite[Claim 7.1]{kociumaka2022small}}]\label{claim:same-period}
    Let $U$ and $V$ be strings such that $U$ is a prefix of $V$, and $|V|\le 2|U|$.
    If $U$ is $d$-mismatch periodic with $d$-mismatch period $Q$, then $V$ either is not $d$-mismatch periodic or has $d$-mismatch period $Q$.
\end{claim}

Charalampopoulos, Kociumaka, and Wellnitz~\cite{charalampopoulos2020faster} showed that the set of $k$-mismatch occurrences has a very regular structure:

\begin{fact}[{See~\cite[Section 3]{charalampopoulos2020faster}}]\label{fact:ham-k-occ}
Let $P$ and $T$ be two strings such that $|P|\le |T| \le 3/2|P|$.
\begin{enumerate}
\item If $P$ is not $k$-mismatch periodic, then there are $O(k)$ $k$-mismatch occurrences of $P$ in $T$.
\item If $P$ is $k$-mismatch periodic with period $Q$, then any two $k$-mismatch occurrences $i\le i'$ of $P$ in $T$ satisfy $i\equiv i' \pmod{|Q|}$ and $\hd{T(i-|P|\dd i']}{Q^\infty}\le 3k$.
\end{enumerate}
\end{fact}

They also presented efficient offline algorithms for computing the $k$-mismatch period and the $k$-mismatch occurrences in the so-called PILLAR model. In this model, one is given a family of strings $\mathcal{X}$ for preprocessing. The elementary objects are fragments $X[i\dd j]$ of strings $X \in \mathcal{X}$. Given elementary objects $S, S_1, S_2$, the PILLAR operations are:
\begin{enumerate}
\item $\mathsf{Access}(S,i)$: Assuming $i \in [1\dd |S|]$, retrieve $S[i]$.
\item $\mathsf{Length}(S)$: Retrieve the length $|S|$ of $S$.
\item $\lcp(S_1,S_2)$: Compute the length of the longest common prefix of $S_1$ and $S_2$.
\item $\lcp^R(S_1,S_2)$: Compute the length of the longest common suffix of $S_1$ and $S_2$.
\item $\mathsf{IPM}(S_1,S_2)$: Assuming that $|S_2| \le 2|S_1|$, compute the set of the starting positions
of occurrences of $S_1$ in $S_2$, which by Fine and Wilf periodicity lemma~\cite{fine1965uniqueness} can be represented as one arithmetic progression.
\end{enumerate}
In the read-only model, operations $\mathsf{Access}$ and $\mathsf{Length}$ can be implemented in constant time and $O(\log m)$ bits. The operations $\lcp$ and $\lcp^R$ can be implemented naively via character-by-character comparison in $O(\min\{|S_1|,|S_2|\})$ total time and $O(\log m)$ bits. Finally, the $\mathsf{IPM}$ operation can be implemented in $O(|S_1|+|S_2|)$ total time and $O(\log m)$ bits (see e.g.~\cite{rytter2003maximal}).

As a corollary, we immediately obtain:

\begin{corollary}[{From~\cite[Lemma 4.4]{charalampopoulos2020faster}}]\label{lemma:compute-ham-period}
Given random access to a string $U$, testing whether it is $d$-mismatch periodic, and, if so, computing its $d$-mismatch period, can be done using $O(d|U|)$ time and $\Oh(d)$ space.
\end{corollary}

\subsection{Read-only algorithm for the pattern matching with \texorpdfstring{\boldmath $k$}{k} mismatches}

The above implementation of the PILLAR operations further implies an offline algorithm that finds all $k$-mismatch occurrences of $P$ in $T$ in $\Ot(k^2 \cdot |T|)$ time and $\Ot(k^2)$ space (see \cite[Main Theorem~8]{charalampopoulos2020faster}).
Nevertheless, we provide a more efficient online algorithm that additionally provides the mismatch information for every $k$-mismatch occurrence of $P$.

\begin{theorem}\label{thm:ro-km-algo}
    There is a deterministic online algorithm that finds all $k$-mismatch occurrences of a length-$m$ pattern $P$ within a text $T$ using $\Oh(k\log m)$ space and $\Oh(k\log m)$ worst-case time per character. The algorithm outputs the mismatch information along with every reported $k$-mismatch occurrence of $P$.
\end{theorem}

Consistently with the streaming algorithm of~\cite{clifford2019streaming}, our algorithm uses a family of exponentially-growing
prefixes to filter out candidate positions. However, in order to use the structural properties of \cref{fact:ham-k-occ} efficiently, we construct a different family $\Pp$ to ensure that we are either working in an approximately
periodic region of the text or processing an aperiodic prefix.

We first add to $\Pp$ the prefixes $R_j = P[\dd \min\{m, \lfloor(3/2)^j\rfloor\}]$
for $j\in [0\dd \lceil \log_{3/2} m \rceil]$.
If $R_j$ is $k$-mismatch periodic but $R_{j+1}$ is not,
we also add to $\Pp$ the shortest extension of $R_j$ that is not $k$-mismatch periodic.
Hereafter, let $\Pp = (P_j)_{j=1}^t$ denote the resulting sequence of prefixes,
sorted in order of increasing lengths, and let $\ell_j = |P_j|$ for every $j\in [1\dd t]$.

\begin{restatable}{claim}{rosqpref}\label{claim:ro-sq-pref}
   The sequence $\Pp=(P_j)_{j=1}^t$ satisfies the following properties:
    \begin{enumerate}[(a)]
        \item\label{claim:enum-corner} $P_1 = P[1]$ and $P_t = P$,
        \item\label{claim:enum-size} $t=|\Pp| = O(\log m)$,
        \item\label{claim:enum-length} for every $j\in [1\dd t)$, we have $\ell_{j+1} \le 3\ell_j/2$,
        \item\label{claim:enum-close}  for every $j\in [1\dd t)$, if $P_j$ is $k$-mismatch periodic with period $Q_j$,
        then $\hd{P_{j+1}}{Q_j^\infty} \le 2k+1$.
    \end{enumerate}
\end{restatable}
\begin{claimproof} Properties~(\ref{claim:enum-corner}), (\ref{claim:enum-size}), and (\ref{claim:enum-length}) are straightforward.
    For Property~(\ref{claim:enum-close}), there are two possible cases:
    if $P_{j+1}$ is $k$-mismatch periodic, \cref{claim:same-period} implies that $P_{j+1}$ has the same $k$-mismatch period $Q_j$ as $P_j$, that is $\hd{P_{j+1}}{Q_j^\infty} \le 2k$.
    Otherwise, by construction,  $P_{j+1}$ is the shortest extension of $P_j$ that is not $k$-mismatch periodic.
    By minimality, removing its last character yields a $k$-mismatch periodic prefix,
    and by~\cref{claim:same-period}, it has the same $k$-mismatch period $Q_j$
    as $P_j$, i.e., we have $\hd{P[\dd \ell_{j+1})}{Q_j^\infty} \le 2k$ for $i < \ell_j$.
    Adding one more character to $P[\dd \ell_{j+1})$ can increase the Hamming distance by at most one.
\end{claimproof}

\subparagraph*{Processing the pattern.}
In the preprocessing phase, we build $\Pp$ and, for each $k$-mismatch periodic prefix $P_j\in \Pp\setminus\{P\}$,
we also retrieve the period $Q_j$ (represented as a fragment of $P_j$) and the mismatch information $\MI(P_{j+1},Q_j^\infty)$.
For subsequent indices $j\in [0\dd \lceil \log_{3/2} m\rceil]$, we add the prefix $R_j$ to $\Pp$.
If $R_j\ne P$, we apply \cref{lemma:compute-ham-period} to test whether
$R_j$ is $k$-mismatch periodic and, if so, retrieve the period~$Q$.
If~$R_j$ is $k$-mismatch periodic, we build $\MI(R_j, Q^\infty)$ and extend $R_j$ while maintaining the mismatch information with the appropriate prefix of~$Q^\infty$.
We proceed until we reach length $|R_{j+1}|$ or $2k+1$ mismatches, whichever comes first.
We add the obtained extension $R_j'$ to $\Pp$ and store the mismatch information $\MI(R_j',Q^\infty)$.
If $\hd{R_j'}{Q^\infty}\le 2k$, then $R_j'=R_{j+1}$ is $k$-mismatch periodic with the same period $Q$.
Otherwise, by \cref{claim:same-period}, neither $R_j'$ nor $R_{j+1}$ are $k$-mismatch periodic.
Processing each~$j$ takes $\Oh(|R_{j+1}|k)$ time and $\Oh(k)$ space, for a total of $\Oh(mk)$ time and $\Oh(k\log m)$ space across $j\in [0\dd \lceil \log_{3/2} m\rceil]$.

\subparagraph*{Processing the text.}
Our online algorithm processing the text $T$ consists of $t=|\Pp|$ layers, each of which reports the $k$-mismatch occurrences of $P_j\in \Pp$,
along with the underlying mismatch information.

The first layer, responsible for $P_1=P[1]$, is implemented naively in $\Oh(1)$ space and time per character.

Each of the subsequent layers receives the $k$-mismatch occurrences of $P_j$ and outputs the $k$-mismatch occurrences of $P_{j+1}$.
The processing is based on the following simple observation:

\begin{observation}\label{obs:ro-km-filter}
    If $P_{j+1}$ has a $k$-mismatch occurrence at position $i$ of $T$, then $P_j$ has a $k$-mismatch occurrence at position $i-\ell_{j+1}+\ell_j$ of $T$.
\end{observation}

We partition $T$ into blocks of length $b:=\lceil \ell_{j}/2 \rceil$ and, for each block $T(rb\dd (r+1)b]$, use a separate subroutine to output $k$-mismatch occurrences of $P_{j+1}$ at positions $i\in (rb\dd (r+1)b]$. This subroutine receives the $k$-mismatch occurrences of $P_{j}$ at positions $i-\ell_{j+1}+\ell_j \in (rb-\ell_{j+1}+\ell_j\dd (r+1)b-\ell_{j+1}+\ell_j]$. It is considered \emph{active} as the algorithm reads $T(rb-\ell_{j+1}+\ell_j\dd (r+1)b]$; since $\ell_{j+1}\le \frac32\ell_j$, at most two subroutines are active at any given time. The implementation of the subroutine depends on whether $P_j$ is $k$-mismatch periodic or not.

\subparagraph*{\boldmath $P_j$ is not $k$-mismatch periodic.}
In this case, for every received $k$-mismatch occurrence~$i'$ of~$P_j$, the subroutine stores the mismatch information $\MI(T(i'-\ell_j\dd i'],P_j)$ and, as
the algorithm receives subsequent characters $T[i]$ for $i\in (i'\dd i'+\ell_{j+1}-\ell_j]$, we maintain $\MI(T(i'-\ell_j\allowbreak \dd i], P[\dd \ell_j+i-i'])$ as long as there are at most $k$ mismatches. If this is still the case for $i=i'+\ell_{j+1}-\ell_j$, we report a $k$-mismatch occurrence of $P_{j+1}$ and output $\MI(T(i'-\ell_j\allowbreak \dd i], P[\dd \ell_j+i-i'])=\MI(T(i-\ell_{j+1}\dd i], P_{j+1})$. By \cref{obs:ro-km-filter}, no $k$-mismatch occurrence of $P_{j+1}$ is missed.
Moreover, \cref{fact:ham-k-occ} guarantees that the subroutine receives $\Oh(k)$ $k$-mismatch occurrences of $P_j$, and thus it uses $\Oh(k)$ space and $\Oh(k)$ time per character.

\subparagraph*{\boldmath $P_j$ is $k$-mismatch periodic with period $Q_j$.}
In this case, we wait for the leftmost $k$-mismatch occurrence $p\in (rb-\ell_{j+1}+\ell_j\dd (r+1)b-\ell_{j+1}+\ell_j]$ of $P_j$ and ignore all the subsequent occurrences of $P_j$. We use the received mismatch information $\MI(T(p-\ell_j\dd p], P_j)$ and the preprocessed mismatch information $\MI(P_{j+1},Q_j^\infty)$ to construct $\MI(T(p-\ell_j\dd p], Q_j^\infty)$; by the triangle inequality, the size of this set is guaranteed to be at most $3k$.
As the algorithm receives subsequent characters of $T[i]$ for $i\in (p\dd (r+1)b]$, we maintain $\MI(T(p-\ell_j\dd i], Q_j^\infty)$ as long as the number of mismatches does not exceed $6k+1$. Whenever $i\ge p+\ell_{j+1}-\ell_j$ and $i\equiv p+\ell_{j+1}-\ell_j \pmod{|Q_j|}$, we extract $\MI(T(i-\ell_{j+1}\dd i], Q_j^\infty)$
from $\MI(T(p-\ell_j\dd i], Q_j^\infty)$ and use the precomputed mismatch information $\MI(P_{j+1},Q_j^\infty)$ to construct $\MI(T(i-\ell_{j+1}\dd i], P_j)$. If it is of size at most $k$, we report $i$ as a $k$-mismatch occurrence of $P_j$.

As for the correctness, we argue that we miss no $k$-mismatch occurrence $i\in (rb\dd (r+1)b]$ of $P_{j+1}$ in $T$.
Since $\hd{T(i-\ell_{j+1}\dd i]}{P_{j+1}}\le k$ and $\hd{P_{j+1}}{Q_j^\infty}\le 2k+1$, we have $\hd{T(i-\ell_{j+1}\dd i]}{Q_{j}^\infty}\le 3k+1$.
Moreover, by \cref{obs:ro-km-filter}, $i-\ell_{j+1}+\ell_j$ is a $k$-mismatch occurrence of $P_j$.
\Cref{fact:ham-k-occ} further implies that $i-\ell_{j+1}+\ell_j \equiv p \pmod{|Q_j|}$ and $\hd{T(p-\ell_{j}\dd i-\ell_{j+1}]}{Q_j^\infty}\le 3k$.
Consequently, $\hd{T(p-\ell_{j}\dd i]}{Q_j^\infty}\le 6k+1$, and thus we compute $\MI(T(i-\ell_{j+1}\dd i], Q_j^\infty)$
and report $i$ as a $k$-mismatch occurrence of $P_{j+1}$.

We conclude with the complexity analysis: the working space is $\Oh(k)$, dominated by the maintained mismatch information.
Moreover, whenever we compute $\MI(T(i-\ell_{j+1}\dd i], P_j)$, the size of this set is, by the triangle inequality, at most $6k+1+2k+1\le 8k+2$,
and it can be computed in $\Oh(k)$ time.

\subparagraph*{Summary.}
Overall, each subroutine of each level takes $\Oh(k)$ space and $\Oh(k)$ time per character. Since there are $t=\Oh(\log m)$ levels and each level contains at most two active subroutines, the algorithm takes $\Oh(k\log m)$ space and $\Oh(k\log m)$ time per text character.
Although our pattern preprocessing algorithm is an offline procedure, we can run it while the algorithm reads the first $m/2$ characters of the text.
Then, while the algorithm reads further $m/2$ characters, it can process two characters at a time to catch up with the input stream. This does not result in any delay on the output because the leftmost $k$-mismatch occurrence of~$P$ is at position $m$ or larger.

\subsection{Read-only algorithm for \texorpdfstring{\boldmath \kpalh}{k-LHD-PAL}}\label{sec:ro-pal-ham}
\begin{restatable}{theorem}{thmropalham}\label{thm:ro-pal-ham}
    There is a deterministic online algorithm that solves the \kpalh{} problem for a string of length $n$
    using $\Oh(k\log n)$ space and $\Oh(k\log n)$ worst-case time per character.
\end{restatable}

The algorithm uses a filtering approach to select positions where a prefix close to \PAL can end.
Define a family $\Pp = \{P_j = T[\dd \lfloor (3/2)^j\rfloor] : j\in [1\dd \lfloor \log_{3/2} n \rfloor]\}$ of prefixes of the text, and let $\ell_j = |P_j|$,
setting $\ell_0=0$ for notational convenience.

\begin{restatable}{claim}{palroham}\label{claim:correct_pal_ro_ham}
Consider $j\in [1\dd \lfloor \log_{3/2} n \rfloor]$ and a position $i\in (2\ell_{j-1}\dd 2\ell_j]$.
If $\hd{T[\dd i]}{\PAL} \le k$, then $i$ is a $2k$-mismatch occurrence of $P_{j}^R$ in $T$.
Moreover, $\hd{T[\dd i]}{\PAL}=\hd{T(i-i'\dd i]}{P_{j}[1\dd i')^R}$ for $i'=\lfloor i/2 \rfloor$.
\end{restatable}
\begin{claimproof}
    Note that $i>2\ell_{j-1}\ge \ell_j$ implies that $P_{j}$ is a prefix of $T[\dd i]$ and, equivalently, $P_{j}^R$ is a suffix of $T[\dd i]^R$. \Cref{prop:kpalham} implies $2\cdot \hd{T[\dd i]}{\PAL} = \hd{T[\dd i]}{T[\dd i]^R}\ge \hd{T(i-\ell_{j}\dd i]}{P_{j}}$.
    Thus, if $\hd{T[\dd i]}{\PAL} \le k$, then $i$ is a $2k$-mismatch occurrence of $P_{j}$ in $T$.
    Since $T[\dd i']$ is a prefix of~$P_{j}$, \cref{prop:kpalham} further implies $\hd{T[\dd i]}{\PAL}=\hd{T(i-i'\dd i]}{T[\dd i']^R}=\hd{T(i-i'\dd i]}{P_{j}[1\dd i')^R}$.  
\end{claimproof}

The algorithm constructs the family $\Pp$ as it reads the text. For each level $j$, we implement a subroutine responsible for positions $i\in (2\ell_{j-1}\dd 2\ell_j]$.
First, while reading $T[\ell_{j}\dd 2\ell_{j-1})$, we launch the pattern-matching algorithm of \cref{thm:ro-km-algo} in order to compute the $2k$-mismatch occurrences of $P_j^R$ in $T_j = T[\dd 2\ell_{j})$ and feed the pattern-matching algorithm with the pattern $P_j$ and a prefix $T[\dd 2\ell_{j-1})$ of $T_j$, ignoring any output produced.
The total number of characters provided is $\ell_j + 2\ell_{j-1} \le 7\cdot (2\ell_{j-1}-\ell_j)$, so we can feed the algorithm with $\Oh(1)$ characters for every scanned character of $T$.
Then, while reading $T[2\ell_{j-1}\dd 2\ell_{j})$, we feed the pattern-matching algorithm with subsequent characters of $T$.
For every reported $2k$-mismatch occurrence $i$ of $P_j^R$ in $T_j$, we retrieve the mismatch information $\MI(T(i-\ell_j\dd i], P_j^R)$ and obtain 
$\MI(T(i-i'\dd i], P_j[\dd i']^R)$ by removing the entries corresponding to the leftmost $\ell_j-i'$ positions.
We report the size of this set (or $\infty$ if the size exceeds $k$) as $\hdk{T[\dd i]}{\PAL}{k}$.

By \cref{claim:correct_pal_ro_ham}, all positions $i\in (2\ell_{j-1}\dd 2\ell_j]$ such that $\hd{T[\dd i]}{\PAL} \le k$ pass the test and the distance $\hd{T[\dd i]}{\PAL}$ is equal to the size of the set $\MI(T(i-i'\dd i], P_j[\dd i']^R)$.
As for the complexity analysis, observe that, for each level $j$, the pattern-matching algorithm uses~$\Oh(k\cdot j)$ space and takes $\Oh(k\cdot j)$ time per character. Since, at any time, there is a constant number of active levels, the main algorithm uses $\Oh(k\log n)$ space and takes $\Oh(k\log n)$ time per character.

\subsection{Read-only algorithm for \texorpdfstring{\boldmath \ksqh}{k-LHD-SQ}}\label{sec:ro-sq-ham}
\begin{restatable}{theorem}{thmrosqham}\label{thm:ro-sq-ham}
    There is a deterministic online algorithm that solves the \ksqh{} problem for a string $T \in \Sigma^n$
    using $\Oh(k \log n)$ space and $\Oh(k\log n)$ worst-case time per character.
\end{restatable}

Our algorithm is very similar to the pattern-matching algorithm of \cref{thm:ro-km-algo}.
We use the same sequence $\Pp = (P_j)_{j=1}^t$ of prefixes, now defined for $P=T$.
Again, we set $\ell_j=|P_j|$ for $j\in [1\dd t]$.
Instead of \cref{obs:ro-km-filter}, we use \cref{obs:ksqh} to argue that our filtering procedure is correct.

\subparagraph*{\boldmath Processing $\Pp$.}
We build $\Pp$ in an online fashion so that the prefix $P_j$ is constructed while scanning $T(\ell_j\dd \lceil 3\ell_{j}/2\rceil ]$.
If $P_j$ is $k$-mismatch periodic, then we also identify $P_{j+1}$ and build $\MI(P_{j+1},Q_j^\infty)$.

For subsequent indices $j\in [0\dd \lfloor{\log_{3/2} n}\rfloor]$, we add the prefix $R_j$ to $\Pp$ as soon as it has been read.
Then, we launch an offline procedure that applies \cref{lemma:compute-ham-period} to test whether~$R_j$ is $k$-mismatch periodic and, if so, retrieves the period $Q$.
If~$R_j$ is $k$-mismatch periodic, we build $\MI(R_j, Q^\infty)$ and extend $R_j$ while maintaining the mismatch information with the appropriate prefix of~$Q^\infty$.
We proceed until we reach length $|R_{j+1}|$ or $2k+1$ mismatches, whichever comes first.
We add the obtained extension $R_j'$ to $\Pp$ and store the mismatch information $\MI(R_j',Q^\infty)$.
If $\hd{R_j'}{Q^\infty}\le 2k$, then $R_j'=R_{j+1}$ is $k$-mismatch periodic with the same period $Q$.
Otherwise, by \cref{claim:same-period}, neither $R_j'$ nor $R_{j+1}$ are $k$-mismatch periodic.
Processing each $j$ takes $\Oh(|R_{j+1}|k)$ time and $\Oh(k)$ space, and this computation needs to be completed while the algorithm reads $T(|R_j|\dd |R_{j+1}|]$. This gives $\Oh(k)$ time per position since $\lfloor \frac32 |R_{j}|\rfloor \le |R_{j+1}|\le \lceil \frac32 |R_{j}|\rceil$.

Across all indices $j\in [0\dd \lfloor{\log_{3/2} n}\rfloor]$, the preprocessing algorithm takes $\Oh(k)$ space and time per character (since no two indices are processed simultaneously).

\subparagraph*{Computing the distances.}
For each level $j\in [1\dd t]$, we implement a subroutine responsible for even positions $i\in [2\ell_j\dd 2\ell_{j+1})$;
this procedure is active as we read $T[\ell_j\dd 2\ell_{j+1})$.
As described above, the pattern $P_j$ is identified while the algorithm reads $T(\ell_j\dd \lceil 3\ell_{j}/2\rceil ]$
and, if~$P_j$ is $k$-mismatch periodic, the period $Q_j$ and the mismatch information $\MI(P_{j+1},Q_j^\infty)$ are also computed at that time.
While reading $T[\lceil 3\ell_{j}/2\rceil\dd 2\ell_j)$, we launch the pattern-matching algorithm of \cref{thm:ro-km-algo} to report the $k$-mismatch occurrences of $P_j$ in $T_j=T[\dd \ell_{j}+\ell_{j+1})$ and feed this algorithm with the pattern $P_j$ and the prefix $T[\dd 2\ell_j)$ of the text~$T_j$.
The total number of characters provided is $3\ell_j \le 6\cdot \frac12 \ell_j$, so can feed the pattern-matching algorithm with $\Oh(1)$ character for every scanned character of $T$.
Then, while reading $T[2\ell_j\dd \ell_j+\ell_{j+1})$, we feed the pattern-matching algorithm subsequent text characters. For every $i'\in [2\ell_j\dd \ell_j+\ell_{j+1})$, we learn whether $i'$ is a $k$-mismatch occurrence of $P_j$ and, if so, we obtain the mismatch information $\MI(P_j,T(i'-\ell_j\dd i'])$.
How we utilise this output depends on whether $P_j$ is $k$-mismatch periodic or not:
if $P_j$ is not $k$-mismatch periodic, then~$T_j$ contains $O(k)$ $k$-mismatch occurrences of $P_j$
and storing them explicitly requires little space.
When $P_j$ is $k$-mismatch periodic, $T_j$ must exhibit similar periodicity, which we can use to
avoid storing all occurrences explicitly.

\subparagraph*{\boldmath $P_j$ is not $k$-mismatch periodic.}
In this case, for every received $k$-mismatch occurrence $i'$ of~$P_j$, we store the mismatch information $\MI(T(i'-\ell_j\dd i'],P_j)$ and, as
the algorithm receives subsequent characters $T[i]$ for $i\in (i'\dd 2(i'-\ell_j)]$, we maintain $\MI(T(i'-\ell_j\dd i], T[\dd \ell_j+i-i'])$ as long as there are at most $k$ mismatches. If this is still the case for $i=2(i'-\ell_j)$, we report that $T[\dd i]$ is a $k$-mismatch square,
with $\hd{T[\dd i]}{\SQ}=\hd{T(i'-\ell_j\dd i]}{T[\dd \ell_j+i-i']}=\hd{T(i/2\dd i]}{T[\dd i/2]}$.
By \cref{obs:ksqh}, no $k$-mismatch square $T[\dd i]$ is missed.
Moreover, \cref{fact:ham-k-occ} guarantees that there are $\Oh(k)$ $k$-mismatch occurrences of $P_j$, and thus we use $\Oh(k)$ space and $\Oh(k)$ time per character to process all of them.

\subparagraph*{\boldmath $P_j$ is $k$-mismatch periodic with period $Q_j$.}
In this case, we wait for the leftmost $k$-mismatch occurrence $p\in[2\ell_j\dd \ell_j+\ell_{j+1})$ of $P_j$ and ignore all the subsequent occurrences of~$P_j$. We use the received mismatch information $\MI(T(p-\ell_j\dd p], P_j)$ and the preprocessed mismatch information $\MI(P_{j+1},Q_j^\infty)$ to construct $\MI(T(p-\ell_j\dd p], Q_j^\infty)$; by the triangle inequality, the size of this set is guaranteed to be at most~$3k$.
As the algorithm receives subsequent characters of $T[i]$ for $i\in (p\dd 2\ell_{j+1})$, we maintain $\MI(T(p-\ell_j\dd i], Q_j^\infty)$ as long as the number of mismatches does not exceed $6k+1$. Whenever $i/2\ge p-\ell_j$ and $i/2 \equiv p-\ell_j \pmod{|Q_j|}$, we extract $\MI(T(i/2\dd i], Q_j^\infty)$
from $\MI(T(p-\ell_j\dd i], Q_j^\infty)$ and use the precomputed mismatch information $\MI(P_{j+1},Q_j^\infty)$ to construct $\MI(T[\dd i/2], Q_j^\infty)$ first, and then derive $\MI(T[\dd i/2],T(i/2\dd i])$. If the latter is of size at most $k$, we report $T[\dd i]$ as a $k$-mismatch square.

As for the correctness, we argue that we miss no $k$-mismatch square $T[\dd i]$ with $i\in (2\ell_j\dd 2\ell_{j+1}]$.
Since $\hd{T(i/2\dd i]}{T[\dd i/2]}\le k$ and $\hd{P_{j+1}}{Q_j^\infty}\le 2k+1$, as a corollary we obtain $\hd{T(i/2\dd i]}{Q_{j}^\infty}\le 3k+1$.
Moreover, by \cref{obs:ksqh}, $i/2+\ell_j$ is a $k$-mismatch occurrence of $P_j$.
\Cref{fact:ham-k-occ} further implies that $i/2+\ell_j \equiv p \pmod{|Q_j|}$ and $\hd{T(p-\ell_{j}\dd i/2]}{Q_j^\infty}\le 3k$.
Consequently, $\hd{T(p-\ell_{j}\dd i]}{Q_j^\infty}\le 6k+1$, and thus we compute  $\MI(T[\dd i/2],T(i/2\dd i])$
and report $T[1\dd i]$ as a $k$-mismatch square.

We conclude with the complexity analysis: the working space is $\Oh(k)$, dominated by the maintained mismatch information.
Moreover, whenever we compute $\MI(T[\dd i/2],T(i/2\dd i])$, the size of this set is, by the triangle inequality, at most $6k+1+2k+1\le 8k+2$,
and it can be computed in $\Oh(k)$ time.

\subparagraph*{Summary.}
Overall, each level takes $\Oh(k\log n)$ space and $\Oh(k\log n)$ time per character, dominated by the pattern-matching algorithm of \cref{thm:ro-km-algo}.
However, since constantly many levels are processed at any given time, the entire algorithm still uses $\Oh(k\log n)$ space and $\Oh(k\log n)$ time per character.

\section{Language Edit Distance problems}
The \emph{edit distance} between two strings $U$ and $V$, denoted by $\edd{U}{V}$, is the minimum number of
character insertions, deletions, and substitutions required to transform $U$ into~$V$. Similar to the Hamming distance, the edit distance from a string $U$ to \PAL and \SQ can be expressed in terms of self-similarity of $U$. This allows us to use similar approaches as for the Language Hamming distance problems, with tools for the Hamming distance replaced with appropriate tools for the edit distance.

By replacing the Hamming distance sketch~\cite{clifford2019streaming} with the edit distance sketch of Bhattacharya and \koucky~\cite{bhattacharya2023locally}.

\begin{restatable}{theorem}{thmpaledit}\label{th:pal_edit}
    There is a randomised streaming algorithm that solves the \kpale{} problem for a string
    of length $n$ using $\Ot(k^2)$ bits of space and $\Ot(k^2)$ time per character.
\end{restatable}

Furthermore, the results of Bhattacharya and \koucky~\cite{bhattacharya2023locally} show a reduction from the edit distance to the Hamming distance via locally consistent string decompositions, which allows reducing the \ksqe{} problem to \ksqh{}, solved via~\cref{rm:hd-sq-mi}:

\begin{restatable}{theorem}{thmsquareedit}\label{th:streaming_ksqe}
    There is a randomised streaming algorithm that solves the \ksqe{} problem
    for a string of length $n$ using $\Ot(k^2)$
    bits of space and $\Ot(k^2)$ time per character.
\end{restatable}

Finally, by replacing the online read-only algorithm for finding the $k$-mismatch occurrences of a pattern in a text with an online read-only algorithm for finding $k$-error occurrences and the structural results for the Hamming distance with the structural results for the edit distance, we obtain algorithms for \kpale{} and \ksqe{}:

\begin{restatable}{theorem}{thmropaledit}\label{thm:ro-pal-e}
    There is a deterministic online read-only algorithm that solves the \kpale{} problem for a string of length $n$
    using $\Ot(k^4)$ bits of space and $\Ot(k^4)$ time per character.
\end{restatable}

\begin{restatable}{theorem}{thmsquaresedit}\label{thm:ro-sq-e}
    There is a deterministic online read-only algorithm that solves the \ksqe{} problem for a string of length $n$
    using $\Ot(k^4)$ bits of space and $\Ot(k^4)$ amortised time per character.
\end{restatable}

\bibliographystyle{plainurl}
\bibliography{main.bib}

\appendix


\section{Edit distance, palindromes, and squares}
The \emph{edit distance} between two strings $U$ and $V$, denoted by $\edd{U}{V}$, is the minimum number of
character insertions, deletions, and substitutions required to transform $U$ into~$V$.
For a formal definition, we first rely on the notion of an \emph{alignment} between fragments of strings.
\begin{definition}[\cite{kociumaka2022small}]
    A sequence $\caA=(u_t,v_t)_{i=0}^{m}$ is an \emph{alignment}
    of $U$ onto $V$ if $(u_0,v_0)=(0,0)$, $(u_{t},v_{t})\in \{(u_{t-1}+1,v_{t-1}+1),(u_{t-1}+1,v_{t-1}),(u_{t-1},v_{t-1}+1)\}$ for $i\in [1\dd m]$, and $(u_m,v_m) =(|U|,|V|)$.
    \begin{itemize}
        \item If $(u_{t},v_{t})=(u_{t-1}+1,v_{t-1})$, we say that $\caA$ \emph{deletes}
            $U[u_t]$,
        \item If $(u_{t},v_{t})=(u_{t-1},v_{t-1}+1)$, we say that $\caA$ \emph{inserts} $V[v_t]$,
        \item If $(u_{t},v_{t})=(u_{t-1}+1,v_{t-1}+1)$, we say that $\caA$ \emph{aligns}
            $U[u_t]$ and $V[v_t]$. If~additionally $U[u_t]=
            V[v_t]$, we say that $\caA$ \emph{matches} $U[u_t]$ and
            $V[v_t]$; otherwise, $\caA$ \emph{substitutes} $V[v_t]$ for
            $U[u_t]$.
    \end{itemize}
\end{definition}
The \emph{cost} of an alignment $\caA$ of $U$ onto $V$, is the total number of characters that $\caA$ inserts, deletes, or substitutes.
Now, we define the edit distance $\edd{U}{V}$ as the minimum cost of an alignment
of $U$ onto~$V$.
An alignment of $U$ onto $V$ is \emph{optimal} if its cost is equal to $\edd{U}{V}$.

A sequence of edits that an alignment $\caA$ uses to transform $U$ into $V$ (specifying the involved positions and characters) is called an \emph{edit sequence} (of the alignment).

\begin{example}
A string $U = \mathtt{ababc}$ can be transformed onto $V = \mathtt{bbac}$ by substituting $U[1]=\mathtt{a}$ for $V[1]=\mathtt{b}$
and deleting $U[4]=\mathtt{b}$. The corresponding alignment is $(0,0)$, $(1,1)$, $(2,2)$, $(3,3)$, $(4,3)$, $(5,4)$.
\end{example}

For an integer $k$, we denote \[\ed_{\le k}(X,Y)=\begin{cases} \edd{X}{Y} & \text{if }\edd{X}{Y}\le k,\\
    \infty & \text{otherwise.}\end{cases}\]
For strings $P, T \in \Sigma^\ast$, we say that a position $i$ is a \textit{$k$-error occurrence} of $P$ in $T$ if $\edd{T(j\dd i]}{P}\le k$ for some $j\in [1\dd i]$.

Similarly to the Hamming distance, we can measure the edit distance from a string $U$ to \PAL and \SQ in terms of the self-similarity of $U$.

\begin{restatable}{property}{propedtopal}\label{prop:ed-to-pal}
    For any string $U\in\Sigma^*$, we have
    \[\edd{U}{\PAL} = \min_i\{\min\{\edd{U[\dd i]}{U[i+1\dd ]^R}, \edd{U[\dd i]}{U[i+2\dd ]^R}\}\}=\tfrac12 \edd{U}{U^R}.\]
\end{restatable}
\begin{proof}
    First, consider an optimum alignment $(u_t,u'_t)_{t=0}^m$ of $U$ onto $U^R$.
    For every $t\in [0\dd m]$, the alignment maps $U[\dd u_t]$ onto $U^R[\dd u'_t]$
    and $U(u_t \dd ]$ onto $U^R(u'_t\dd ]$. In particular, \[\edd{U}{U^R}=\edd{U[\dd u_t]}{U^R[\dd u'_t]}+ \edd{U(u_t \dd ]}{U^R(u'_t\dd ]}.\]
   
    Suppose that there exists an index $t\in [0\dd m]$ such that $u_t+u'_t=|U|$.
    In this case, $U^R[\dd u'_t] = U(u_t\dd ]^R$ and $U^R(u'_t\dd ]=U[\dd u_t]^R$.
    Consequently,
    \begin{align*}
        \edd{U}{U^R} &=\edd{U[\dd u_t]}{U^R[ dd u'_t]}+\edd{U(u_t\dd ]}{U^R(u'_t\dd ]} \\
        &=\edd{U[\dd u_t]}{U(u_t \dd]^R}+\edd{U(u_t\dd ]}{U[\dd u_t]^R} \\
        &= 2\cdot \edd{U[\dd u_t]}{U(u_t \dd]^R},
    \end{align*}
    that is, $\edd{U}{U^R}=\edd{U[\dd i]}{U[i+1\dd ]^R}$ holds for $i=u_t$.

    Otherwise, as the sequence $(u_t+u'_t)_{t=0}^m$ increases from $0$ to $2|U|$,
    there exists $t\in [1\dd m]$ such that $u_{t-1}+u'_{t-1} < |U| < u_{t}+u'_{t}$.
    In particular,  $(u_{t},u'_{t})=(u_{t-1}+1,u'_{t-1}+1)$, so $U^R[\dd u'_{t-1}] = U(u_t\dd ]^R$, $U^R(u'_{t-1} \dd u'_{t}]=U(u_{t-1}\dd u_t]$, and $U^R(u'_t\dd ]=U[\dd u_{t-1}]$.
    Consequently,
        \begin{align*}
        &\edd{U}{U^R} \\ &=\edd{U[\dd u_{t-1}]}{U^R[ \dd u'_{t-1}]}+\edd{U(u_{t-1}\dd u_t]}{U^R(u'_{t-1}\dd u'_t]}+\edd{U(u_t\dd ]}{U^R(u'_t\dd ]} \\
        &=\edd{U[\dd u_{t-1}]}{U(u_t \dd]^R}+\edd{U(u_{t-1}\dd u_t]}{U(u_{t-1}\dd u_t]^R}+\edd{U(u_t\dd ]}{U[\dd u_{t-1}]^R} \\
        &= 2\cdot \edd{U[\dd u_{t-1}]}{U(u_t \dd]^R},
    \end{align*}
    that is, $\edd{U}{U^R}=\edd{U[\dd i]}{U[i+2\dd ]^R}$ holds for $i=u_{t-1}=u_t-1$.

    This completes the proof that
    \begin{equation}\label{eq:prop-pal-ed-1}
        2\min_i\{\min\{\edd{U[\dd i]}{U[i+1\dd ]^R}, \edd{U[\dd i]}{U[i+2\dd ]^R}\}\}\le \edd{U}{U^R}.
    \end{equation}

    Next, consider a palindrome $V$ such that $\edd{U}{\PAL}=\edd{U}{V}$.
    The triangle inequality implies
    \begin{equation}\label{eq:prop-pal-ed-2}
        \edd{U}{U^R} \le \edd{U}{V}+\edd{V}{U^R}=\edd{U}{V}+\edd{V^R}{U^R}=2\edd{U}{V}=2\edd{U}{\PAL}.
    \end{equation}
   
    Finally, consider an index $i\in [0\dd |U|]$
    and note that \[\edd{U}{\PAL}\le \edd{U}{U[i+1\dd ]^R \cdot U[i+1\dd ]} \le \edd{U[i+1\dd ]^R}{U[\dd i]}.\]
    Similarly, if $i\in [0\dd |U|)$, then
    \[\edd{U}{\PAL}\le \edd{U}{U[i+2\dd ]^R\cdot  U[i+1] \cdot U[i+2\dd ]} \le \edd{U[i+2\dd ]^R}{U[\dd i]}.\]
    Consequently,
    \begin{equation}\label{eq:prop-pal-ed-3}
        \edd{U}{\PAL} \le \min_i\{\min\{\edd{U[\dd i]}{U[i+1\dd ]^R}, \edd{U[\dd i]}{U[i+2\dd ]^R}\}\}.
    \end{equation}

    Combining \eqref{eq:prop-pal-ed-1}, \eqref{eq:prop-pal-ed-2}, and~\eqref{eq:prop-pal-ed-3}
    yields the result.
\end{proof}

\begin{corollary}\label{cor:ro-pal-edit}
    Let $U$ be a string of length $n$ and $m = \lfloor{n/2\rfloor}$. We have
    $$\ed_{\le k} (U, \PAL) = \min_{i\in [m-k\dd m+k]}\{\min\{\edk{U[\dd i]}{U[i+1\dd ]^R}{k}, \edk{U[\dd i]}{U[i+2\dd ]^R}{k}\}$$
\end{corollary}
\begin{proof}
    The edit distance between two strings is at least the difference of their lengths,
    hence we only need to consider strings that differ by at most $k$ in length.
\end{proof}

\begin{restatable}{property}{propsqed}\label{prop:sq-ed}
Let $U \in \Sigma^n$. We have $\edd{U}{\SQ} = \min_{i}\{\edd{U[\dd i]}{U[i+1\dd ]}\}$.
\end{restatable}
\begin{proof}
    First, for any $i \in [1\dd n]$, we have $\edd{U}{\SQ} \le \edd{U[\dd i]}{U[i+1\dd ]}$,
    as editing the substring of $U$ equal to $U[\dd i]$ into $U[i+1\dd ]$ yields
    the string $U[i+1\dd ]U[i+1\dd ]$, which is a square.
    Now, let $V^2$ be a square such that $\edd{U}{\SQ} = \edd{U}{V^2}$.
    Let $r$ be the position of the rightmost character of $U$ that was not deleted,
    and whose position in $V^2$ (after applying the edits) is at most~$|V|$.
    If it is not defined, we put $r = 0$.
    We then have $\edd{U}{V^2} = \edd{U[\dd r]}{V} + \edd{U[r+1\dd ]}{V}$.
    Applying the triangle inequality,
    we obtain $\edd{U}{\SQ} = \edd{U}{V^2} \ge \edd{U[\dd r]}{U[r+1\dd ]}$,
    hence $\edd{U}{\SQ} \ge \min_{i \in [1\dd n]}\{\edd{U[\dd i]}{U[i+1\dd ]}\}$.
\end{proof}

\begin{corollary}\label{cor:k-sq-ed}
    If $U \in \Sigma^n$, then \[\ed_{\le k}(U,\SQ)
    = \min_{i \in [\lfloor n/2 \rfloor-k\dd \lfloor n/2 \rfloor+k]} \min \ed_{\le k}(U[\dd i], U(i\dd]).\]
\end{corollary}

\section{Streaming algorithms for the LED problems}\label{sec:streaming-LED}
In this section, we present streaming algorithms for \kpale and \ksqe.

\subsection{Locally consistent decompositions and sketches}
\newcommand{\eval}{\mathsf{eval}}
\newcommand{\GG}{\mathbf{G}}
Our algorithms use locally consistent string decompositions and edit distance sketches of Bhattacharya and \koucky~\cite{bhattacharya2023locally,editpm}.
\subparagraph*{Locally consistent string decompositions.} The randomised
decomposition algorithm of Bhattacharya and  \koucky~\cite{bhattacharya2023locally} receives two integers $k, n$ as input and starts by selecting two families of hash functions of size $\lambda = O(\log n)$.\footnote{\label{notation}In~\cite{bhattacharya2023locally}, $\lambda$ is denoted by $L$, $\tau$ by $T$, and $\mu$ by $M$. We changed the notation to avoid collisions.} Then, when it receives a string $U$ of length at most $n$ as input, the algorithm outputs a sequence $\GG(U)$ of context-free grammars $\GG(U) = G_1^U\cdots G_s^U$ with certain properties,
called run-length straight-line programs or RLSLP for short (the exact definition of an RLSLP is not important for this work; an interested reader can refer to~\cite{DBLP:conf/mfcs/NishimotoIIBT16}). We denote the length of the sequence $s$ by $|\GG(U)|$ and extend the notations for indexing and concatenating strings (e.g. sequences of characters) to sequences of grammars in a natural way.
Each grammar $\GG(U)[i]$ output by the algorithm represents a \emph{unique} string, denoted by $\eval(\GG(U)[i])$.\
We extend $\eval$ into a morphism over grammar sequences: for grammars $G_1,\ldots,G_s$
let $\eval(G_1\cdots G_s) = \eval(G_1)\cdots \eval(G_s)$.

The decomposition satisfies the following properties:
\begin{fact}[{\cite[Theorem 3.1]{bhattacharya2023locally}}]\label{fact:grammar-decomp}
    Let $U, V$ be a pair of strings of length at most $n$ such that $\edd{U}{V} \le k$.
    Let $\GG(U)$ and $\GG(V)$ be the sequences of
    grammars output by the decomposition algorithm on inputs $U$ and $V$ respectively using the same choice of random hash functions. The following is true for $n$ large enough:
    \begin{enumerate}[(a)]
        \item With probability at least $1-2/n$, $U = \eval(\GG(U))$ and $V = \eval(\GG(V))$ ;
        \item With probability at least $1-2/\sqrt{n}$, for all $i,j$ the grammars $\GG(U)[i]$ and $\GG(V)[j]$ have size $\Ot(k)$;
        \item \label{item:k-different} With probability at least $0.9$, $|\GG(U)| = |\GG(V)|$, $\GG(U)[i] = \GG(V)[i]$ for every $1 \le i\le \GG(U)$ except at most $k$ of them,
        and $\edd{U}{V} = \sum_{i=1}^s \edd{\eval(\GG(V)[i])}{\eval(\GG(U)[i])}$.
    \end{enumerate}
\end{fact}

The decomposition can be updated efficiently upon appending or prepending a character to a string. Let $\tau = O(\log n \log^* n)$ be a parameter defined as in~\cite{bhattacharya2023locally}. (See Footnote~\ref{notation}).
While it may happen that $\GG(UV)$ is shorter than $\GG(U)$,
these decompositions are related as follows:
\begin{corollary}[{Of~\cite[Lemmas 4.1, 4.2]{bhattacharya2023locally}}]\label{cor:small-decomp-change}
     Consider $U, V\in\Sigma^*$ such that $|U|+|V| \le n$.
     Let $G = \GG(U)$, $G' = \GG(UV)$, $G'' = \GG(VU)$, and $s = |G|$, $s' = |G'|$, $s'' = |G''|$.
     We have:
     \begin{enumerate}
         \item $G[\dd s-\tau] = G'[\dd s-\tau ]$ and $G(\tau \dd] = G''(s''-s+\tau \dd]$,
         \item $|U| \le |\eval(G'[\dd \min\{s+\tau, s'\}])|$ and
         $|U| \le |\eval(G''(\max\{s''-s+\tau, 0\} \dd])|$.
     \end{enumerate}
\end{corollary}
In particular, this result implies that if the index of a grammar in $\GG(U)$
is less than $|\GG(U)|-\tau$ (resp. more than $\tau$),
then appending (resp. prepending) characters to $U$ will not change that grammar.
Extending the terminology~\cite{bhattacharya2023locally},
we call a grammar that remains unchanged when appending (resp. prepending)
any string a \textit{right-committed grammar} (resp. \textit{left-committed grammar}),
while the others are \textit{right-active} (resp. \textit{left-active}) \textit{grammars}.
When a grammar is both left- and right-committed, we simply say that it is \emph{committed}.
In what follows, we use $\alpha_r(\GG(U))$ (resp. $\alpha_l(\GG(U))$) to denote
the number of right-active (resp. left-active) grammars in $\GG(U)$.
\cref{cor:small-decomp-change} implies that $\alpha_l(\GG(U)), \alpha_r(\GG(U)) \le \tau$.

Theorem 5.1 in~\cite{bhattacharya2023locally} presents an algorithm only for the case of appending a character,
essentially the same  can be used to prepend one:
\begin{corollary}[{Of \cite[Theorem 5.1]{bhattacharya2023locally}}]\label{cor:update-decomp}
    Let $U\in\Sigma^{\le n-1}$, $G=\GG(U)$, and $s=|G|$.
    There are two algorithms \textsf{AppendCharacter} and \textsf{PrependCharacter}
    that run in $\Ot(k)$ time and:
    \begin{enumerate}
        \item Given a character $a\in\Sigma$ and $G[\max(1, s-\tau)\dd s]$, the algorithm \textsf{AppendCharacter}
        outputs $G'$ such that $\GG(Ua) = G[1 \dd \max(1, s-\tau)) G'$ and $|G'|\le 4\tau\lambda$.
        \item Given a character $a\in\Sigma$ and $G[1 \dd \min\{\tau+1, s\}]$, the algorithm \textsf{PrependCharacter}
        outputs $G'$ such that $\GG(aU) = G' \cdot G(\min\{\tau+1, s\}\dd s]$ and $|G'| \le 4\tau \lambda$.
    \end{enumerate}  
\end{corollary}
It follows from the above two results that,
by storing the last $\alpha_r(\GG(U))+\tau$ (resp. first $\alpha_l(\GG(U))+\tau$) grammars of $\GG(U)$,
we can compute the last $\alpha_r(\GG(UV))+\tau$ grammars
of $\GG(UV)$ (resp. first $\alpha_l(\GG(VU))+\tau$ grammars of $\GG(VU)$)
by applying the algorithm of \cref{cor:update-decomp} $|V|$ times.
These grammars include the right-active grammars of $\GG(UV)$
(resp. left-active grammars of $\GG(VU)$).

\subparagraph*{Edit distance sketches. }
\newcommand{\genc}[1]{\enc( #1 )}
Our algorithms exploit an extension of the \textit{edit distance sketches} of Bhattacharya and \koucky~\cite{bhattacharya2023locally}.

\begin{fact}[{\cite[Lemma 3.13]{bhattacharya2023locally}}]\label{prop-of-genc}
Let $\mu = \Ot(k)$ be a parameter. (See Footnote~\ref{notation}). There is a mapping $\enc$ from the set of grammars output by the decomposition algorithm to the set of strings of length $\mu$ on a polynomial size that guarantees that the following is satisfied:
\begin{enumerate}
\item A grammar can be encoded and decoded in $O(\mu)$ time;
\item Encodings of two equal grammars are equal;
\item Encodings of two distinct grammars output by the decomposition algorithm differ in all $\mu$ characters with probability at least $1-2\mu/n$.
\end{enumerate}
\end{fact}
We extend this encoding to a sequence of grammars: $\genc{G_1\cdots G_s} = \genc{G_1} \cdots \genc{G_s}$.

\begin{lemma}\label{lemma:edit-sketches}
    There is an edit distance sketch $\skek$ that uses $\Ot(k^2)$ bits
    of space, and supports the following operations assuming that
    all involved strings have length at most $n$:
    \begin{enumerate}[(a)]
        \item \textbf{Append:} given $\skek(U)$ and a character $a\in\Sigma$,
        compute $\skek(Ua)$,
        \item \textbf{Prepend:} given $\skek(U)$ and a character $a\in\Sigma$,
        compute $\skek(aU)$,
        \item \textbf{Distance:} given $\skek(U)$ and $\skek(V)$,
        compute $\ed_{\le k} (U,V)$.
    \end{enumerate}
    All three operations take $\Ot(k^2)$ time and space and
    errs with probability inverse polynomial in $n$.
\end{lemma}
\begin{proof}
\newcommand{\skkmu}{\sk_{k \cdot \mu}}
Let $\GG(U)$ be the locally consistent decomposition of a string $U$.
In what follows, we denote $\alpha_r = \alpha_r(\GG(U))$ and $\alpha_l = \alpha_l(\GG(U))$.
If $s = |\GG(U)| \le 4\tau$, $\skek(U) = \genc{\GG(U)}$ and takes $\Ot(k)$ space.
Otherwise, $\skek(U)$ is defined as a tuple
\[(\GG_l, \Ss, \GG_r) = \left(\GG(U)[1\dd s_l],
\sk_{k \cdot \mu}(\genc{\GG(U)(s_l \dd s-s_r]}), \GG(U)(s-s_r \dd s]\right),\]
where $s_l$ and $s_r$ are integers satisfying the invariant
\begin{equation}\label{eq:edsk-invariant}
 2\tau \ge s_l \ge \alpha_l + \tau
\text{ and }2\tau \ge s_r \ge \alpha_r + \tau.
 \end{equation}
Note that this definition does not define $\skek(U)$ uniquely
(as there are multiple possible choices of $s_l$ and $s_r$),
but the invariant of \cref{eq:edsk-invariant} ensures that
we can always perform updates and distance operations using the sketch. Initially (i.e., the first time $s$ becomes larger than $4\tau$), we set $s_l = s_r = 2\tau$, which satisfy the invariant of \cref{eq:edsk-invariant} by \cref{cor:small-decomp-change}.

\textbf{Update operations.}
Let $(\GG_l, \Ss, \GG_r)$ be the sketch of $U$, $s_l = |\GG_l|$ and $s_r = |\GG_r|$.
Given a character $a\in\Sigma$, we obtain the sketch of $Ua$
feeding $(\GG_r[s_r-\tau \dd s_r], a)$
into the algorithm \textsf{AppendCharacter} (Corollary~\ref{cor:update-decomp}).
Let $G_{out}$ be the sequence of grammars output by this algorithm
and let $G' = \GG_r[\dd s_r-\tau)G_{out}$
be the sequence of grammars obtained by replacing
$\GG_r[s_r-\tau \dd s_r]$ with $G_{out}$
in $\GG_r$.
If $|G'| \le 2\tau$, we output $(\GG_l, \Ss, G')$
for the sketch of $Ua$.
Otherwise, we only store explicitly the last $2\tau$ grammars of $G'$,
and add the others at the end of the Hamming sketch $\Ss$.
More formally, we construct the Hamming distance sketch $\Ss'$
by concatenating $\Ss$ and $\skkmu(\genc{G'[\dd |G'|-2\tau]})$.
We then output $\skek(Ua) = (\GG_l, \Ss', G'(|G'|-2\tau\dd])$.
By \cref{cor:update-decomp}, $|G'| \le 4\tau\lambda + \tau = \Ot(1)$,
hence we can build $\skkmu(\genc{})$ in time $\Ot(k^2)$ using \cref{fact:ham_sketch},
and the subsequent sketch concatenation also takes $\Ot(k^2)$ time.
In both cases, \cref{cor:small-decomp-change} ensures that
the invariant of \cref{eq:edsk-invariant} is satisfied.

Prepending a character works similarly using $\GG_l$ and the \textsf{PrependCharacter}
algorithm of \cref{cor:update-decomp}.

\textbf{Distance operation.}
Let $U, V $ be strings satisfying the conditions of \cref{fact:grammar-decomp},
and consider the sketches $\skek(U) = (\GG_l, \Ss, \GG_r)$
and $\skek(V) = (\GG_l', \Ss', \GG_r')$ obtained from decompositions
such that the properties of \cref{fact:grammar-decomp} hold.
By \cref{fact:grammar-decomp}(\ref{item:k-different}),
at most $k$ grammars of the decompositions of $U$ and $V$ differ:
we can recover them from the Hamming sketches of $\GG(U)$ and $\GG(V)$.
We can obtain $\skkmu(\genc{\GG(U)})$ sketches by
concatenating $\skkmu(\genc{\GG_l})$, $\Ss$ and $\skkmu(\genc{\GG_r})$, which takes $\Ot(k^2)$ time
(and similarly to obtain $\skkmu(\genc{\GG(V))}$).
Notice that the above operation results in $\skkmu(\genc{\GG(U)})$
regardless of the length of $|\GG_l|$ and $|\GG_r|$:
this normalization helps us avoid issues related to the definition of the sketches.
Using these sketches, we can recover $\genc{G_i}$
for every $i$ such that $\GG(U)[i] \neq \GG(V)[i]$.
We can then compute the edit distance between $U$ and $V$ using \cref{fact:grammar-decomp}(\ref{item:k-different}):
\begin{align*}\label{eq:grammar-ed-sum}
\edd{U}{V} &= \\
&=\sum_{i=1}^s \edd{\eval(\GG(U)[i])}{\eval(\GG(V)[i])}\\
& = \sum_{i: \GG(U)[i] \neq \GG(V)[i]} \edd{\eval(\GG(U)[i])}{\eval(\GG(V)[i])}   
\end{align*}
Given two RLSLPs $G, G'$ of size at most $m$ and $d = \edd{\eval(G)}{\eval(G')}$,
one can compute $\min\{d, k+1\}$ in time $\Ot(m + k^2)$
using the algorithm of Ganesh et al.~\cite{ganesh2022compression}.
Using binary search over $k$, we can reduce this running time to $\Ot(m + d^2)$
when $d \le k$.
Therefore, we can compute the sum above in time $\Ot(k^2)$:
there are $k$ RLSLPs of size $\Ot(k)$, and we can stop the computation as
soon as the sum of $d_i = \edd{\eval(\GG(U)[i])}{\eval(\GG(V)[i])}$ exceeds $k$. In total, this costs
$\Ot(\mu + k^2 + \sum_i \mu + d_i^2 \le k\mu + \left(\sum_i d_i\right)^2) = \Ot(k^2)$ time.
The space complexity can be upper-bounded (up to polylog factors) by the time complexity.

\textbf{Probability of success.}
The results of \cref{fact:grammar-decomp} hold with constant probability, which can be boosted in a standard way by repeating the scheme a logarithmic number of times.
\end{proof}

\subsection{A streaming algorithm for \texorpdfstring{\kpale}{k-LED-PAL}}\label{sec:kpale}
\thmpaledit*

Our algorithm maintains two sketches $x, x'$ initialised to sketches of the
empty string $\ske_{2k}(\eps)$.
Upon receiving $T[i]$, we append it to $x$ and prepend it to $x'$.
After performing this update, we have $x = \ske_{2k}(T[\dd i])$,
and $x' = \ske_{2k}(T[\dd i]^R)$.
We then use them to compute $\ed_{\le 2k+1} (T[\dd i],T[\dd i]^R)$,
which gives us $\ed_{\le k}(T[\dd i],\PAL)$ from Property~\ref{prop:ed-to-pal}.
By \cref{lemma:edit-sketches}, updating the sketches and using them to compute
the edit distance takes $\Ot(k^2)$ time per character and $\Ot(k^2)$ bits of space.

\subsection{A streaming algorithm for \texorpdfstring{\ksqe}{k-LED-SQ}}\label{sec:ksqe}

\thmsquareedit*

Let $U = T[\dd i]$ be such that $\edd{U}{\SQ} \le k$.
By Property~\ref{prop:sq-ed}, there exist strings $V,W$ such that
$U = VW$ and $\edd{V}{W} \leq k$.
Let $G = \GG(V)$ and $G' = \GG(W)$.
Under the assumptions of Fact~\ref{fact:grammar-decomp},
we have $|G| = |G'| = s$ and $G$ and $G'$ differ at at most $k$ indices,
hence the Hamming distance between $\enc(G)\enc(G')$ and \SQ
is at most $k\mu = \Ot(k^2)$.
Therefore, our idea is to use the algorithm of \cref{th:square_ham} with parameter $K = k\mu$
to detect approximate squares in the stream of encoded grammars,
and use this information to find positions in $i$ in the text
where $\edd{T[\dd i]}{\SQ} \le k$.
We first discuss the reduction to \ksqh,
and then describe how to use it to solve \ksqe.

\subparagraph*{Reduction to \ksqh.}
Two issues arise for the reduction to \ksqh{}: the grammar decomposition
of $U$ is not necessarily equal to the concatenation of $G$ and $G'$,
and adding a character to the text may reduce the size of the  decomposition,
hence the conversion of the stream of characters to a stream of grammars is not trivial.

Let $G'' = \GG(U)$ and $q = |G''|$.
We assume w.l.o.g. that $q \ge 8\tau$ (which implies $s \ge 2\tau$);
otherwise, we store all the grammars in $G''$ explicitly (taking $\Oh(\tau\mu) = \Ot(k)$ bits of space), and do not use the following reduction to \ksqh to find all the differing grammars, but use exhaustive pairwise comparison.

While $G''$ is not necessarily equal to $GG'$, it follows from \cref{cor:small-decomp-change}
that $G''[1\dd s-\tau] = G[1\dd {s-\tau}]$, $G''(q-s-\tau\dd q] = G'({\tau}\dd s]$ and $2s \le q+2\tau$.
In other words, $G''$ is equal to $GG'$ with at most $4\tau = \Ot(1)$ modified, inserted or deleted grammars.
Furthermore, while appending a character to the text may reduce the number of grammars in its decomposition,
right-committed grammars will not change when appending a letter.
Therefore, we can build a monotone stream of grammars by including only right-committed grammars, i.e., grammars whose index has been at $\tau$ positions away from the end of the decomposition at some point.

We decompose $G''$ as follows (see Figure~\ref{fig:ksqe} for an illustration):
\begin{align*}
	 A &= G''[1\dd\tau]
	 &D &= G''(q-s+\tau\dd q-\tau]\\
	 B &= G''(\tau\dd s-\tau]
	 &E &= G''(q-\tau \dd q]\\
	 C &= G''(s-\tau \dd q-s+\tau]&&
\end{align*}
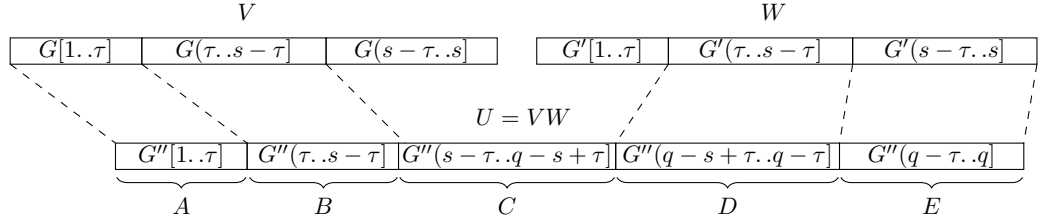
\begin{figure}[htbp]
    \tikzmath{
    \wvar = 5; 
    \wxpos= 20;
    \wypos= 0;
    \uxpos= 4;
    \uypos= 4;
    }
    \centering
    \begin{tikzpicture}[scale=0.35, every node/.style={scale=.9}]
    	\node (v) at (9, 2) {$V$};
        \draw (0,     0) rectangle ++(\wvar, 1) node[midway] { $G[1\dd\tau]$ };
        \draw (\wvar, 0) rectangle ++(\wvar+2, 1) node[midway] { $G(\tau\dd s-\tau]$ };
        \draw (2*\wvar+2, 0) rectangle ++(\wvar+1.5, 1) node[midway] { $G(s-\tau \dd s]$ };

    	\node (w) at (9+\wxpos, 2) {$W$};
        \draw (0+\wxpos,     0-\wypos) rectangle ++(\wvar, 1) node[midway] { $G'[1\dd \tau]$ };
        \draw (\wvar+\wxpos, 0-\wypos) rectangle ++(\wvar+2, 1) node[midway] { $G'(\tau\dd s-\tau]$ };
        \draw (2*\wvar+2+\wxpos, 0-\wypos) rectangle ++(\wvar+2, 1) node[midway] { $G'(s-\tau \dd s]$ };

    	\node (u) at (15.5+\uxpos, 2-\uypos) {$U = VW$};
        \draw (\uxpos,     0-\uypos) rectangle ++(\wvar, 1) node[midway] { $G''[1\dd\tau]$ };
        \draw (\wvar+\uxpos, 0-\uypos) rectangle ++(\wvar+0.75, 1) node[midway] { $G''(\tau \dd s-\tau]$ };
        \draw (2*\wvar+0.75+\uxpos, 0-\uypos) rectangle ++(\wvar+3.25, 1) node[midway] { $G''(s-\tau\dd q-s+\tau]$ };

        \draw (3*\wvar+4+\uxpos, 0-\uypos) rectangle ++(\wvar+3.5, 1) node[midway] { $G''(q-s+\tau\dd q-\tau]$ };
        \draw (4*\wvar+7.5+\uxpos, 0-\uypos) rectangle ++(\wvar+2, 1) node[midway] { $G''(q-\tau \dd q]$ };
       
        \draw [decorate,decoration={brace,amplitude=4pt}] (\uxpos-0.05+\wvar,-\uypos-0.3) -- ++(-\wvar+0.1, 0) node [below,midway,yshift=-5pt] {$A$};
        \draw [decorate,decoration={brace,amplitude=4pt}] (\uxpos-1.3+2*\wvar+2,-\uypos-0.3) -- ++(-\wvar-0.6, 0) node [below,midway,yshift=-5pt] {$B$};
        \draw [decorate,decoration={brace,amplitude=4pt}] (\uxpos-0.05+3*\wvar+4,-\uypos-0.3) -- ++(-\wvar-3.15, 0) node [below,midway,yshift=-5pt] {$C$};
        \draw [decorate,decoration={brace,amplitude=4pt}] (\uxpos-0.05+4*\wvar+7.5,-\uypos-0.3) -- ++(-\wvar-3.4, 0) node [below,midway,yshift=-5pt] {$D$};
        \draw [decorate,decoration={brace,amplitude=4pt}] (\uxpos-0.05+5*\wvar+9.5,-\uypos-0.3) -- ++(-\wvar-1.9, 0) node [below,midway,yshift=-5pt] {$E$};
        \draw[dashed] (0,0)  -- (\uxpos, -\uypos+1);
        \draw[dashed] (\wvar, 0)  -- ++(\uxpos, -\uypos+1);
        \draw[dashed] (2*\wvar+2, 0)  -- ++(\uxpos-1, -\uypos+1);
        \draw[dashed] (\wxpos+\wvar,0)  -- ++(\uxpos-\wvar-1, -\uypos+1);
        \draw[dashed] (\wxpos+2*\wvar+2,0)  -- ++(\uxpos-\wvar+0.5, -\uypos+1);
        \draw[dashed] (\wxpos+3*\wvar+4,0)  -- ++(\uxpos-\wvar+0.5, -\uypos+1);
    \end{tikzpicture}
    \caption{Decomposition of $U = VW$.}
    \label{fig:ksqe}
\end{figure}
By~\cref{cor:small-decomp-change}, $A = G[1\dd\tau]$, $B = G(\tau\dd s-\tau]$, $D=G'(\tau\dd s-\tau]$, and $E = G'(s-\tau\dd s]$. Furthermore, by~\cref{fact:grammar-decomp}, $\hd{B}{D} \le k$, where the Hamming distance between two equal-length sequences of grammars is defined as the number of positions where they differ. $A$ and~$E$ are stored explicitly, which takes $O(\tau \cdot \mu) = \Ot(k)$ space.

Further, for $c \in [1\dd 4\tau]$, define $\hat{A_c}$ as follows:
\[\hat{A_c} =\begin{cases}
	A(\tau-c\dd ], &\text{ if } c\le \tau,\\
	\#^{c-\tau}A, &\text{ otherwise,}	
\end{cases}\]
where $\#$ denotes the dummy (empty) grammar.

\begin{claim}\label{claim:sq-stream-reduction}
	If $\edd{U}{\SQ} \le k$,
	then there exists $c\le 4\tau$ such that
	$\hd{\hat{A_c}B}{CD} \le k+\tau$.
\end{claim}
\begin{proof}
	For $c = |C| \le 4\tau$,
	we have $|\hat{A_c}| = c$,
	and $\hd{\hat{A_c}B}{CD} = \hd{\hat{A_c}}{C} + \hd{B}{D} \le \tau+k$.
\end{proof}

For each $c\in [1\dd 4\tau]$, we create a stream $T_c = \hat{A_c}BCD$ from a stream of $ABCD$ (by either dropping the first $\tau-c+1$ values if $c \le \tau$, or by inserting $c-\tau$ dummy grammars at the start otherwise) and apply the algorithm for \ksqh (\cref{th:square_ham}) with
parameter $K = (k+\tau)\mu = \Ot(k^2)$ to it, which returns every position where $\hd{\hat{A_c}B}{CD} \le k+\tau$. We use this information as a filter before testing whether $\edd{U}{\SQ} \le k$.

Finally, notice that $E$ contains $\tau$ grammars, which may be more than
just the right-active grammars, as the size of the decomposition may shrink after appending
a character.
We only want grammars up to index $q-\tau$ to be in the stream of right-committed RLSLPs,
i.e., we need to \textit{remove} grammars from the stream when the size of the decomposition
is reduced.
To circumvent this issue, we simply store the grammars output by the algorithm for each of the $\tau$ latest updates.
When the grammar decomposition shrinks, we simply go back to the corresponding grammars output and proceed from there.
This takes $\tau\cdot \Ot(k^2) = \Ot(k^2)$ bits of space.

\subparagraph*{Computing the edit distance between $V$ and $W$.}
Using \cref{rm:hd-sq-mi} and \cref{prop-of-genc}, we can also retrieve from
the algorithm for \ksqh the list of mismatches between $\genc{\hat{A_c}BCD}$
and the closest square, from which we can extract the mismatching grammars in ${\hat{A_c}B}$ and ${CD}$.
As we store $A$ explicitly, we can assume to know every grammar of $C$.

We would like to compute $d = \ed_{\le k} (V,W)$.
Under the assumption of \cref{fact:grammar-decomp},
we have
\begin{align*}
\edd{V}{W} &= \\
&=\sum_{i=1}^s \edd{\eval(G[i])}{\eval(G'[i])}\\
&=\edd{\eval(G[\dd \tau])}{\eval(G'[\dd \tau])}\\
&=\sum_{j: B[j] \neq D[j]} \edd{\eval(B[j])}{\eval(D[j])} + \edd{ \eval(G(s-\tau \dd ])}{\eval(G'(s-\tau \dd ])}.
\end{align*}
As for any strings $X,X',Y,Y'$ we have $\edd{XY}{X'Y'} \le \edd{X}{X'}+\edd{Y}{Y'}$,
we can combine any two terms of the above sum by concatenating
the corresponding grammars in the $\eval$ function. Therefore,
$$\edd{V}{W} = \edd{\eval(EA)}{\eval(G(s-\tau \dd s]G'[1 \dd \tau])}
		+ \sum_{j: B[j] \neq D[j]} \edd{\eval(B[j])}{\eval(D[j])}.$$ Finally, note that $\eval(G[s-\tau+1\dd s]G'[1\dd \tau]) = \eval(C)$ (see~\cref{fig:ksqe}), and therefore $\edd{V}{W} = \edd{\eval(EA)}{\eval(C)}
		+ \sum_{j: B[j] \neq D[j]} \edd{\eval(B[j])}{\eval(D[j])}$.

Both $\eval(EA)$ and $\eval(C)$ can be represented by an RLSLP of size $\Ot(k)$ obtained by concatenating the RLSLPs $G[s-\tau+1],\ldots, G[s], G'[1],\ldots, G[\tau]$ and the RLSLPs in $C$ respectively. We can now apply the algorithm of Ganesh et al.~\cite{ganesh2022compression} to compute each of the terms in the sum above in $\Ot(k^2)$ time and space.

\subparagraph*{Summary. }
Our algorithm runs $\Ot(1)$ copies of the algorithm of Theorem~\ref{th:square_ham}
in parallel, with distance parameter $K = \Ot(k^2)$: this uses $\Ot(k^2)$ bits of space.
Upon receiving the $i$th character of the text, we use the algorithm of \cref{cor:update-decomp}
to update the decomposition of the text, which takes time $\Ot(k)$.
We add all but the last $\tau$ grammars to the stream of right-committed grammars.
The algorithm of Corollary~\ref{cor:update-decomp} may add
up to $4\tau\lambda = \Ot(1)$ grammars to the decomposition of the text,
hence we commit at most $\Ot(1)$ grammars per character of the text.
Note that we only need to know whether the stream
of right-committed grammars is close to a square after pushing all symbols
in the encoding of the grammars, and do not need the intermediate results.
Therefore, we can use the algorithm of \cref{rm:hd-sq-mi}
to push the at most $4\tau\lambda\mu = \Ot(k)$ symbols
in the encoding of the grammars into the \ksqh subroutines
using $\Ot(k\sqrt{K} + K) =\Ot(k^2)$ time.
After adding all the new right-committed grammars, if any copy of the \ksqh algorithm signals that the current stream of grammars is within distance $K$ of \SQ,
we run the above algorithm to compute the distance between $T[\dd i]$ and \SQ, which costs $\Ot(k^{2})$ time and space.

Finally, the algorithm explicitly stores the RLSLPs in $A$, the last $\tau$ grammars
and the last $\tau$ output of each \ksqh subroutine,
which uses $\tau\mu + \tau\mu + 2\tau\mu + \tau\Ot(k^2) = \Ot(k^2)$ bits of space in total.

Overall, our algorithm uses $\Ot(k^2)$ time per character and $\Ot(k^2)$ bits of space.

\section{Read-only algorithms for the LED problems}\label{sec:ro-edit}
In this section, we present read-only algorithms for \kpale and \ksqe.

\subsection{Structure of \texorpdfstring{$k$}{k}-error occurrences}
\subparagraph*{Approximate periodicity.}
\begin{definition}[{\cite{charalampopoulos2020faster}}]\label{def:k-error-periodic}
A string $U$ is \textit{$d$-error periodic} if there exists a \textit{primitive} string $Q$ such that $|Q| \leq |U|/128d$ and $\edd{U}{Q^\infty}\le 2d$. Such a string $Q$ is called the \emph{$d$-error period} of $U$.\footnote{The original definition of~\cite{charalampopoulos2020faster} considered the distance between $U$ and a substring of $Q^\infty$. Changing it does not break the structure of $k$-error occurrences, but is important for the read-only algorithms for \kpale and \ksqe.}
\end{definition}

Similarly to the Hamming distance, the condition $|Q| \leq |U|/128d$ implies that if $U$ is $d$-error periodic with $d$-error period $Q$, then $Q$ is equal to some substring of $U$. Furthermore, we exploit the following properties:

\begin{fact}[{\cite{kociumaka2022small}}]\label{fact:prefix-equal-per}
Suppose that a string $X$ is a prefix of a string $Y$, where $|X| < |Y| \le 2|X|$. If $X$ is $k$-error periodic with $k$-error period $Q$, then either $Y$ is not $k$-error periodic, or $Y$ is $k$-error periodic with $k$-error period $Q$.
\end{fact}

\begin{restatable}{lemma}{errorperiodalgo}\label{lemma:compute-edit-period}
    Given random access to a string $U$,
    testing whether it is $d$-error periodic, and in the relevant case computing its $d$-error period,
    can be done using $O(|U|d^2)$ time and~$O(d)$ space.
\end{restatable}
\begin{proof}
    Consider a partition of $U$ into $4d$ substrings $S_1,\ldots, S_{4d}$
    of equal length $|U|/{4d}$. At most $2d$ of them can contain an edit operation,
    hence the others must be equal to a substring of $Q^\infty$,
    and as $|Q| \le |U|/128d = |S_i|/32$, these $S_i$ are (exactly) periodic.

    This observation leads to the following algorithm:
    For $i \in [1\dd 4d]$, we test whether $S_i$ is (exactly) periodic:
    if it is, we compute its period $Q_i'$; otherwise, we go to the next $i$.
    Let $p_i$ denote the ending position of $S_i$ in $U$.
    Otherwise, we construct the string $S_i'$ by extending $S_i$ to the left
    with copies of $Q_i'$,
    and test whether there exists an index $j$ such that the edit distance
    between $U[\dd p_i]$ and $S_i'[j\dd ]$ is at most $2d$.
    Let $Q_i$ be the forward cyclic rotation of $Q_i'$ by $j \mod |Q_i'|$ positions:
    it is a candidate for the $d$-error period of $U$.
    We then test whether the edit distance between $U$ and $(Q_i')^\infty$
    is at most $2d$: in the affirmative case, $Q_i$ is a $d$-error period of $U$;
    otherwise, we go to the next $i$.
    If $U$ is $d$-error periodic, at least one of the first $2d+1$ tests
    must return the $d$-error period.

    Computing the exact period of a string $R$ can be done in linear time and logarithmic space,
    using the algorithm of Rytter~\cite{rytter2003maximal} to find the first occurrence of $R$ in $R[2\dd ]R$.
    Computing the index that minimizes the edit distance upper bounded by $d$
    can be done in $O(|U|d)$ time and~$O(d)$ space using
    the classical dynamic programming algorithm for the edit distance~\cite{ukkonen1983approximate}.
    We process $O(d)$ substrings $S_i$ sequentially,
    hence our algorithm takes $O(|U|d^2)$ time and~$O(d)$ space.
\end{proof}

\subparagraph*{Structure of approximate occurrences.}

\begin{definition}[Chain of $k$-error occurrences]\label{def:ed-chain}
 A sequence $\Cc = p_1 < \cdots < p_\ell$ of~$k$-error occurrences of $P$ in $T$ forms a \textit{chain} if the following two conditions are satisfied:
\begin{enumerate}
\item There exists an integer $q$ such that $p_1, \ldots, p_\ell$ is an arithmetic progression with difference $q$;
\item There exists an integer $k'\le k$ such that for every $p_i$, $\min_{j \le p_i}$ we have $\edd{P}{T[j\dd p_i]} = k'$.
\end{enumerate}  
\end{definition}

\begin{fact}[{From~\cite[Theorem 5.1, Claim 5.16, Claim 5.17]{charalampopoulos2020faster}, see also~\cite[Corollary III.5]{kociumaka2022small}}]\label{fact:ed-k-occ}
    Let $P, T$ be two strings such that $|T| \le 3/2|P|$ and
    $T$ starts and ends with a $k$-error occurrence of $P$. Then one of the following holds:
    \begin{enumerate}
        \item Either there are~$O(k^2)$ $k$-error occurrences of $P$ in $T$, or
        \item There is a primitive string $Q$ such that $P$ is $2k$-error periodic with $2k$-error period $Q$, $\edd{T}{Q^\infty} \leq 6k$, and the occurrences of $P$ in $T$ can be decomposed into $O(k^3)$ chains. For each chain, its difference equals $|Q|$ and the first position in it is within distance $10k$ from a multiple of $|Q|$.
    \end{enumerate}
\end{fact}

\subsection{Online deterministic read-only algorithm for finding \texorpdfstring{$k$}{k}-error occurrences}

\begin{fact}[\cite{landau1988fast}]\label{fact:LV}
Given two strings $U,V$. There is a data structure that can be built using~$O(k^2)$ \lcp queries and occupies $O(k^2)$ space and allows retrieving, for any two prefixes $U',V'$ of $U,V$ respectively, $\ed_{\le k}(U',V')$ in $O(k)$ time.
\end{fact}

To be able to utilise this fact, we show how to answer $\lcp$ queries in $k$-error periodic strings:
\begin{lemma}\label{lm:lcp_ed}
Consider two strings $U,V, Q \in \Sigma^*$. Assume that $Q$ is a primitive string and that there exist $O(k)$-length edit sequences $\ES_U, \ES_V$ between $U$ and $Q^\infty[1 \dd |U|]$ and $V$ and $Q^\infty[1 \dd |V|]$, respectively. There is a read-only algorithm that receives as an input $U,V$ and $\ES_U, \ES_V$ and computes $\lcp(U[i\dd ],V[j\dd ])$ in time $O(k|Q|)$ and space $\Ot(1)$.
\end{lemma}
\begin{proof}
We consider two cases: (1) Both $U[i]$ and $V[j]$ are unedited by the edit sequences; (2) Either $U[i]$ is edited by $\ES_U$ or $V[j]$ is edited by $\ES_V$.

In the first case, let $i'$ (resp., $j'$) be the length of $U[1\dd i]$ (resp., $V[1\dd j]$) after we apply $\ES_U$ (resp., $\ES_V$) to it. We perform an \lcp query on $Q^\infty[i'\dd]$ and $Q^\infty[j'\dd]$. If $i' = j' \pmod{|Q|}$, the answer is $+\infty$. Otherwise, we compute the answer in a naive way, comparing $Q^\infty[i'\dd]$ and $Q^\infty[j'\dd]$ character-by-character. As $Q$ is primitive, we will find a mismatch after performing $O(|Q|)$ comparisons. Let $\ell = \lcp(Q^\infty[i'\dd], Q^\infty[j'\dd])$. If $U[i \dd \min\{i+\ell-1,|U|\}]$ and $V[j \dd \min\{j+\ell-1,|V|\}]$ are unedited by $\ES_U$ and $\ES_V$, then $\lcp(U[i \dd], V[j\dd]) = \min\{\ell, |U|-i+1, |V|-j+1\}$. Otherwise, if $t_U$ and $t_V$ are the leftmost edited characters in the two strings and $t = \min\{t_U, t_V\}$, we output $t-1+\lcp(U[i+t\dd], V[j+t\dd])$.

Consider now the second case. Assume first that $U[i]$ is deleted or substituted. If $U[i] = V[j]$, we output $1+\lcp(U[i+1\dd], V[j+1\dd])$, and otherwise $0$. The subcase when~$V[j]$ is deleted or substituted is symmetrical. As the parameters $i$ and $j$ are increasing, this case can occur $O(k)$ times.

    An \lcp query in $Q^\infty$ (the first case) is either the last step of the algorithm
    or is followed by a call to the first case of the algorithm.
    As the first case happens $O(k)$ times, the algorithm takes $O(k |Q|)$ time.
\end{proof}

\begin{restatable}{lemma}{onlinekedit}\label{lemma:ro-ke-algo}
There is a deterministic online read-only algorithm that finds all $k$-error occurrences of a length-$m$ pattern $P$ in a text $T$ using $\Ot(k^{4})$ bits of space and $\Ot(k^{4})$ amortised time per character.
\end{restatable}
\begin{proof}
We prove that, for every $d\in \mathbb{Z}_{\ge 0}$, there is an algorithm reporting the $k$-error occurrences with a delay exactly $d$.

If $d \ge \frac{m}{4}$, then we partition the text into disjoint blocks of $b=\lfloor{\frac{m}{4}}\rfloor$ characters. Consider a block $T(r-b\dd r]$.
Having processed $T[r+d-b]$, we use the offline algorithm~\cite[Main Theorem 9]{charalampopoulos2020faster} to retrieve the $k$-error occurrences of $P$ ending within the block $T(r-b\dd r]$, reported as~$\Oh(k^3)$ chains. This costs $\Ot(bk^{4})$ time and uses $\Ot(k^{4})$ space.
For every $i\in (r-b\dd r]$, while processing $T[i+d]$, we check if any of the chains contain an occurrence ending at position $i$ and, if so, report the underlying distance $\min_{j} \ed_{\le k}(P, T(j\dd i])$.

Now, suppose that $k \le d < \frac{m}{4}$. We partition $P$ into a prefix $L=P[1\dd m-4d]$ and a suffix $R=P(m-4d\dd m]$.
We recursively report the occurrences of $L$ with a delay of $5d-k \ge 4d$ and store them in a buffer of size $2k+1$.
In particular, while the algorithm processes $T[i+d]$, we have access to $\min_{j} \ed_{\le k}(L, T(j\dd i'])$ for all $i'\in [i-4d-k\dd i-4d+k]$.
We also partition the text into disjoint blocks of $d$ characters. Consider a block $T(r-d\dd r]$. Having received the entire block, we use the offline algorithm~\cite[Main Theorem 9]{charalampopoulos2020faster} to retrieve the $k$-error occurrences of $R$ ending within the block $T(r-d\dd r]$, reported as $\Oh(k^3)$ chains. This costs $\Ot(dk^{4})$ time and uses $\Ot(k^{4})$ space.
If $R$ is far from periodic, it has $\Oh(k^2)$ occurrences. Otherwise, $R$ has a $2k$-error period $Q$ and, if $p$ and $p'$ are the leftmost and rightmost positions in $T(r-d\dd r]$ where the reported occurrences end, then $T(p-4d-k\dd p']$ is at edit distance $\Oh(k)$ from a substring of $Q^\infty$. In that case, we retrieve the underlying edits as well as the $\Oh(k)$ edits between $R$ and a substring of $Q^\infty$.
In either case, while processing $T[i+d]$, if a $k$-edit occurrence of $R$ ends at position $i$, we compute $\ed_{\le k}(R,T(i'\dd i])$ for all $i'\in [i-4d-k\dd i-4d+k]$ using \cref{fact:LV}. In the non-periodic case, we use the naive $O(d)$-time implementation of \lcp queries. This costs $\Oh(dk^2)$ time per occurrence and $\Oh(k^4)$ amortised per position. Otherwise, we use the $\Oh(|Q|\cdot k)$-time implementation of \lcp queries; see \cref{lm:lcp_ed}. This costs $\Oh(|Q|\cdot k^3)$ time per occurrence and, since there are $\Oh(k)$ occurrences among every $|Q|$ positions, also $\Oh(k^4)$ amortised time per position.
We can combine $\ed_{\le k}(R, T(i'\dd i])$ with $\min_{j} \ed_{\le k}(L, T(j\dd i'])$ (minimizing over $i'\in [i-4d-k\dd i-4d+k]$) to obtain $\min_{j} \ed_{\le k}(P, T(j\dd i])$.

Finally, if $d < k$, we partition $P$ into a prefix $L=P[1\dd m-4k]$ and a suffix $R=P(m-4k\dd m]$.
We recursively report the occurrences of $L$ with a delay of $d+3k > 4d$ and store them in a buffer of size $2k+1$.
In particular, while the algorithm processes $T[i+d]$, we have access to $\min_{j} \ed_{\le k}(L, T(j\dd i'])$ for all $i'\in [i-5k\dd i-3k]$.
While processing $T[i+d]$, we compute $\ed_{\le k}(R,T(i'\dd i])$ for all $i'\in [i-5k\dd i-3k]$
using \cref{fact:LV} and the naive $O(k)$-time implementation of \lcp queries. This costs $\Oh(k^3)$ time per position.
We can combine $\ed_{\le k}(R, T(i'\dd i])$ with $\min_{j} \ed_{\le k}(L, T(j\dd i'])$ (minimizing over $i'\in [i-4d-k\dd i-4d+k]$) to obtain $\min_{j} \ed_{\le k}(P, T(j\dd i])$.

At each level of the recursive algorithm, we use $\Ot(k^4)$ space and amortised time per character. The recursive calls decrease the pattern length and increase the delay at least $4$ times, so the depth of the recursion is $\Oh(\log m)$.
\end{proof}

\subsection{Read-only algorithm for \texorpdfstring{\kpale}{k-LED-PAL}}\label{sec:ro-pal-edit}

\thmropaledit*

Consider a family $\Pp = \{P_j : P_j = T[\dd 2^j], j \in [1\dd \lfloor\log n\rfloor] \}$ of prefixes of the text. For each $j$, define $\ell_j = |P_j|$, $T_j = T[\dd \ell_{j+1})$, and $\occ_j$ to be the set of $2k$-error occurrences of $P_j^R$ in $T_j$ larger or equal to~$\ell_j$. Using \cref{lemma:compute-edit-period}, we can decide whether $P_j$ is $2k$-error periodic and if it is, compute its $2k$-error period $Q_j$ in $O(|P_j|k^2)$ time and $O(k)$ space. Using Ukkonen's algorithm~\cite{ukkonen1983approximate}, we can also compute a $O(k)$-length edit sequence $\ES_j$ from $P_j$ to $(Q_j)^\infty$ in $O(k)$ amortised time per character and $\Ot(k^2)$ space.

\begin{claim}\label{claim:correct_pal_ro_edit}
If $\edd{T[\dd i]}{\PAL} \le k$, then $i\in\occ_{j^*}$ for $j^* = \max\{j: \ell_j \le i\}$.
\end{claim}
\begin{proof}
By Corollary~\ref{cor:ro-pal-edit}, $\edd{T[\dd i]}{T[\dd i]^R} \le 2k$. We also have $\ell_j \le i \le 2\ell_j-1$. Hence, $i \in \occj$.
\end{proof}

Define $i' = \lfloor i/2\rfloor$. By~\cref{cor:ro-pal-edit}, to decide the edit distance between $T[\dd i]$ and \PAL, it suffices to compute $e = \min_{j\in [i'-k\dd i'+k]} \{\min\{\edk{T[\dd j]}{T[j+1\dd ]^R}{k}, \edk{T[\dd j]}{T[j+2\dd ]^R}{k}\}\}$.
We will only compute this value if $i \in \occ'_{j^*} \supseteq \occ_{j^*}$ (the set $\occ'_{j^*}$ to be defined below), and otherwise we will output $k+1$.

For each $j$, we run the online algorithm of \cref{lemma:ro-ke-algo} that uses $\Ot(k^4)$ amortised time per character and $\Ot(k^4)$ space to find the set $\occj$ of $k$-error occurrences of $P_j^R$ in $T_j$. We start it immediately after receiving $T[\ell_j]$ and catch up by processing the first $\ell_j$ characters of $T_j$ at once to be online. If $P_j$ is not $2k$-error periodic, then by~\cref{fact:ed-k-occ}, $|\occj| = O(k^2)$ and we define $\occj' = \occj$. If $P_j$ is $2k$-error periodic, let $p_j$ be the leftmost position in $\occj$, and $r_j$ the largest integer such that $\edd{T[p_j\dd p_j+r_j\cdot |Q_j|]}{(Q_j^R)^\infty} \le 12k$. Define $\occj' = \{p_j+m\cdot|Q_j|+\Delta : 0 \le m \le r_j, |\Delta| \le 10k\}$. By Fact~\ref{fact:ed-k-occ}, $\occj' \supseteq \occj$. The integer $r_j$ can be computed by running an instance of Ukkonen's online algorithm~\cite{ukkonen1983approximate} for $T[p_j+1\dd ]$ and $(Q_j^R)^\infty$, which takes $O(k)$ amortised time per character and $\Ot(k^2)$ space. If $i \le p_j + r_j \cdot |Q_j|$ is the current position, the algorithm also allows extracting an $O(k)$-length edit sequence $\ES_j'$ between $T[p_j\dd i]$ and $(Q_j^R)^\infty$ in $O(k)$ time.

\subparagraph*{Computing the distances. }
Let $i$ be the current position and $j^* = \max\{j: \ell_j \le i\}$. Assume that $i \in \occj'$. If $P_j$ is not $2k$-error periodic, we compute $e$ using $k$ instances of Ukkonen's online algorithm~\cite{ukkonen1983approximate}, which takes
$O(k^2)$ amortised time per character in total and $\Ot(k^2)$ space. Otherwise, the value $e$ is computed as follows. First, if $i = p_j$, we run another instance of Ukkonen's online algorithm~\cite{ukkonen1983approximate} to compute an $O(k)$-length edit sequence $\ES_j''$ between $T(p_j-\ell_j-k\dd p_j]$ and $(Q_j^R)^\infty)$, which must exist as $p_j$ is a $2k$-error occurrence of $P_j^R$ and $P_j$ is $k$-error periodic with period $Q_j$. This takes $O(k)$ amortised time and $\Ot(k^2)$ space. Now, to compute the value $e$, we extract $O(k)$-length edit sequences between $T[1\dd \ell_j+k]$ and $Q_j^\infty$ and between $T(i-\ell_j-k\dd i]$ and $(Q_j^R)^\infty$ in $O(k)$ time from $\ES_j$, $\ES_j'$, and $\ES_j''$, and then apply~\cref{fact:LV} and~\cref{lm:lcp_ed} to compute $e$ in $O(k^3 |Q_j|)$ time and $\Ot(k^2)$ space.

\subparagraph*{Summary.}
The algorithm of \cref{lemma:ro-ke-algo} uses $\Ot(k^4)$ amortised time per character and $\Ot(k^4)$ space. Processing aperiodic prefixes costs $\Ot(k^2)$ amortised time per character and $\Ot(k^2)$ space. To upper bound the complexity of processing periodic prefixes, note that we test $O(k)$ positions out of $|Q_j|$, and therefore use $O(k^4)$ amortised time per character and $\Ot(k^2)$ space.

\subsection{Read-only algorithm for \texorpdfstring{\ksqe}{k-LED-SQ}}
\thmsquaresedit*

We first define a filtering family of prefixes $\Pp$ of the text.
Start by considering the prefixes $R_j = T[\dd \lfloor(3/2)^j\rfloor]$,
$j\in [1\dd \lfloor\log_{3/2} n\rfloor]$.
If $R_j$ is $k$-error periodic but $R_{j+1}$ is not,
add to $\Pp$ the shortest extension of $R_j$ that is not $k$-error periodic. Hereafter, let $\Pp = \{P_j\}$ denote the resulting family of prefixes, sorted in order of increasing lengths. For each $j$, let $\ell_j = |P_j|$, $T_j = T[1\dd \ell_{j+1}+\ell_j]$, and $\occj$ the set of $3k$-error occurrences of $P_j$ in $T_j$. We build $\Pp$ as we read $T$. When $T[\ell_j]$ arrives, we add $R_j$ to $P_j$ and launch an instance of the patter-matching algorithm for $T_j$ to compute $\occj$ (we process the first $\ell_j$ characters of $T_j$ at once to be online). Finally, we apply~\cref{lemma:compute-edit-period} to test $R_j$ for $3k$-error periodicity, which requires $O(k^2)$ amortised time and $\Ot(k)$ space. If $R_j$ is $3k$-periodic, the lemma also allows to compute the period $Q_j$ of $R_j$. We finish by computing the longest $3k$-error periodic extension $R'_j$ of $R_j$: by~\cref{fact:prefix-equal-per}, the $3k$-error period of $R'_j$ equals $Q_j$ and therefore we can compute $R'_j$ by running Ukkonen's algorithm on $T$ and $Q_j$, which requires $O(k)$ amortised time and $\Ot(k^2)$ space. The algorithm also outputs $O(k)$-length edit sequence $\ES_j$ between $R'_j$ and $Q_j^\infty[1\dd |R'_j|]$.

\begin{claim}\label{claim:ksqe-filtering}
If $\edd{T[\dd i]}{\SQ} \le k$, then $i'+\ell_{j^*} \in \occ_{j^*}$, where $i' = \lfloor i/2\rfloor$ and $j^* = \max\{j : \ell_j+\lceil{k/2\rceil} \leq i'\}$. Furthermore, if $P_{j^*}$ is $3k$-periodic, then $\edd{T(i'+\ell_{j^*}\dd i]}{Q_{j^*}^\infty} \le 10k+1$.
\end{claim}
\begin{proof}
By Corollary~\ref{cor:k-sq-ed}, we have $\edd{T[\dd i'+t]}{T[i'+t+1\dd i]} \leq k$
for some $t$, $|t|\leq \lceil k/2\rceil$. As $P_{j^*}$ is a prefix of $T[\dd i'+t]$, we have $\edd{P_{j^*}}{T[i'+t+1\dd i'+t+\ell_{j^*}+\Delta]} \leq k$ for some $-k \le \Delta \le k$ and therefore by the triangle inequality $\edd{P_{j^*}}{T[i'+t+1\dd i'+\ell_{j^*}]} \leq 3k$. It follows that $\ell_{j^*} < i'+\ell_{j^*} \le \ell_{j^*+1}+\ell_{j^*}$
is a $3k$-error occurrence of $P_{j^*}$ and hence $i'+\ell_{j^*} \in  \occ_{j^*}$. To show the second part of the claim, note that $i'+t \le \ell_{j^*+1} + k$. Additionally, by the construction of $\Pp$ and~\cref{fact:prefix-equal-per} we have $\edd{P_{j^*}}{(Q_{j^*})^\infty} \le 7k+1$. Applying the triangle inequality one more time, we obtain that $\edd{T[i'+t+1\dd i]}{(Q_{j^*})^\infty} \le 10k+1$,
as cutting equal-length suffixes of both strings cannot increase the edit distance.
\end{proof}

By Corollary~\ref{cor:k-sq-ed}, to get the edit distance between $T[\dd i]$ and \SQ, it suffices to compute the value $e = \min_{j \in [i'-k\dd i'+k]} \edk{T[\dd j]}{T(j\dd ]}{k}$. We will only compute this value if $i'+\ell_{j^*} \in \occ_{j^*}' \supseteq \occ_{j^*}$ to be defined below, and otherwise we output $\infty$. The set $\occ_{j^*}'$ is defined differently depending on whether $P_j$ is $3k$-error periodic.

If $P_j$ is not $3k$-error periodic, $\occ_{j^*}' = \occ_{j^*}$ and by~\cref{fact:ed-k-occ}, occupies $\Ot(k^2)$ space. Using~\cref{lemma:ro-ke-algo}, $\occ_{j^*}'$ can be computed in $\Ot(k^4)$ amortised time per character and $\Ot(k^4)$ space. If $P_j$ is $3k$-error periodic with a period $|Q_j|$, $\occ_{j^*}'$ is defined to be $\{p_{j^*}+t\cdot |Q_{j^*}| + \Delta : 0 \le t \le r_{j^*}, |\Delta| \le 10 k\}$, where $p_{j^*}$ is the leftmost position in $\occ_{j^*}$ and $r_{j^*}$ is the largest integer such that $\edd{T[p_j\dd p_j+r_j\cdot |Q_j|]}{(Q_j^R)^\infty} \le 2 \cdot (10k+1)$. By~\cref{fact:ed-k-occ} and~\cref{claim:ksqe-filtering}, $\occ_{j^*}' \supseteq \occ_{j^*}$.
The set $\occ_{j^*}'$ is computed as follows in this case: First, we identify $p_j$ using~\cref{lemma:ro-ke-algo}. Then, we launch Ukkonen's online algorithm~\cite{ukkonen1983approximate} to compute $r_j$, which takes $O(k^2)$ amortised time per character and $O(k^2)$ space. 

\subparagraph*{Computing the distances. }
Consider the moment when $T[i]$ arrives.
Let $i' = \lfloor i/2\rfloor$ and $j^* = \max\{j : \ell_j+\lceil{k/2\rceil} \leq i'\}$. If $i'+\ell_{j^*} \notin \occ'_{j^*}$, the algorithm outputs $k+1$. Below we assume $i'+\ell_{j^*} \in \occ'_{j^*}$.

If $P_{j^*}$ is not $3k$-error periodic, we compute $e = \min_{j \in [i'-k\dd i'+k]}  \edk{T[\dd j]}{T(j\dd ]}{k}$ by running $k$ instances of Ukkonen's algorithm~\cite{ukkonen1983approximate}. As $|\occ'_{j^*}| = O(k^2)$, testing all positions in $\occ'_{j^*}$ requires $\Ot(k)$ space and $O(k^4)$ amortised time per character in total.

Otherwise, we compute $e$ as follows. If $i'+\ell_{j^*} = p_{j^*}$, we use one instance of Ukkonen's algorithm~\cite{ukkonen1983approximate} to compute an $O(k)$-length edit sequence $\ES''_{j^*}$ between $T[i' \dd i'+\ell_{j^*}]^R$ and $Q_j^\infty[1\dd \ell_{j^*}]$. For each $j \in [i'-k\dd i'+k]$, we first extract $O(k)$-edit sequences between $T[1\dd j]$ and $(Q_{j^*})^\infty[1\dd j]$ and between $T(j \dd i]$ and $(Q_{j^*})^\infty[1\dd j-i]$ from $\ES_{j^*}$, $\ES'_{j^*}$, and $\ES''_{j^*}$ in $O(k)$ time, and then apply Fact~\ref{fact:LV} and Lemma~\ref{lm:lcp_ed} to compute $\edk{T[\dd j]}{T(j\dd ]}{k}$ using $O(|Q_{j^*}| k^3)$ time and $O(k^2)$ space.  

\subparagraph*{Summary.}
Computing $\Pp$ takes $\Ot(k)$ amortised time per character and $\Ot(k^2)$ space. 
The pattern-matching algorithm of~\cref{lemma:ro-ke-algo} takes $\Ot(k^4)$ amortised time per character and $\Ot(k^4)$ space. Processing aperiodic prefixes requires $O(k^4)$ amortised time per character and $\Ot(k)$ space. Processing periodic prefixes takes $O(k^4)$ time and $O(k^2)$ space as by the definition of $\occj$ we test $O(k)$ positions out of $|Q_j|$.
\end{document}